%% file: main_arxiv.tex
\documentclass[mnsc,nonblindrev]{informs3_hide}
\OneAndAHalfSpacedXI 

\usepackage[english]{babel}
\usepackage[autostyle, english = american]{csquotes}
\MakeOuterQuote{"}
\usepackage[colorlinks,linkcolor=blue, citecolor=blue]{hyperref}

\input{major_revision/macros}

\begin{document}


\RUNAUTHOR{Wang and Jiang}

\RUNTITLE{Plan-Driven Adaptive Bidding}

\TITLE{Plan-Driven Adaptive Bidding for First-Price Auctions with Budget Constraints under Nonstationarity}

\ARTICLEAUTHORS{%
\AUTHOR{$\text{Yige Wang}^{\dag}$, $\text{Jiashuo Jiang}^{\dag}$}

\AFF{\  \\
$\dag~$Department of Industrial Engineering \& Decision Analytics, Hong Kong University of Science and Technology
}
}

\ABSTRACT{%
We study budget pacing in repeated first-price auctions when an advertiser's private-value distributions change over time and the stationary competing-bid distribution is unknown. We ask how a feasible expenditure plan should enter online bid shading, learning, and hard budget control.
We establish a plan-to-performance decomposition for a plan-driven projected-dual policy. The policy uses any feasible expenditure plan as a soft target, learns an unknown stationary competing-bid CDF from thresholds revealed after each auction, and enforces the campaign budget on every sample path. Against a distribution-informed expected-budget fluid benchmark, the uniform-plan reward gap is $O(\sqrt T)+O(\mathcal W_T)$, where $\mathcal W_T$ measures heterogeneity in private-value distributions. With a supplied feasible plan, the global gap decomposes into a one-sided $O(\sqrt{T})$ fixed-plan execution term and a plan-mismatch term bounded by $(b/2a)\PlanErr$. The same analysis provides guarantees for strict and relaxed period-cap comparators, exact recovery of the global benchmark under a specific allowance vector, and separate lower bounds establishing the necessity of the temporal-heterogeneity and Plan Error terms.
An upstream planner can translate forecasts or managerial priorities into a feasible spending trajectory, while the online controller adapts bids using realized thresholds and expenditures. The guarantee is modular: it evaluates the final normalized or projected plan through $\PlanErr$. A specific forecasting model can be linked to the guarantee by establishing how its primitive estimation errors propagate to this plan-quality metric.
}%





\maketitle

\section{Introduction}\label{sec:intro}
\input{major_revision/1_intro}

\section{Plan-Driven Adaptive Bidding: A Dual-Based Learning Framework}\label{sec:main_alg}
\input{major_revision/2_dual_problem_main_algorithm}

\section{Uninformative Regime: Performance without a Predicted Plan}\label{sec:unknown_F}
\input{major_revision/3_uninformative_case}

\section{Informative Regime: Performance with a Predicted Plan}\label{sec:known_F}
\input{major_revision/4_informative_case}

\section{Relaxed Plan Comparators and Operational Flexibility}\label{discussion}
\input{major_revision/5_discussion}

\section{Numerical Evaluation}\label{numerical}
\input{major_revision/6_numerical}

\section{Conclusion}\label{conclusion}
This paper develops a plan-driven interface between offline planning and online bidding in budget-constrained first-price auctions. The plan acts as a soft target in the dual update, while the remaining-budget check enforces pathwise feasibility. Private-value distributions may change over time, and the policy learns an unknown stationary competing-bid CDF from revealed thresholds.

The guarantees separate missing timing information, plan quality, and online execution. Without a predicted plan, the uniform policy has gap $O(\sqrt T)+O(\mathcal W_T)$ to the distribution-informed expected-budget fluid benchmark. With a predicted plan, the gap is at most $O(\sqrt{T})+(b/2a)\PlanErr$. Relaxed period allowances trade pacing discipline for flexibility while retaining the aggregate expected budget. Separate lower bounds show that linear dependence on $\mathcal W_T$ and $\PlanErr$ cannot generally be improved in the corresponding worst-case classes.

The numerical study examines performance across horizons, the role of timing information, and sensitivity to plan misspecification. With the common reference initialization, informative plans often improve on the Uniform Plan and Standard Pacing, although attenuating the forecasted change can reduce this benefit under Abrupt Change. The comparisons also depend on initialization, and finite-horizon performance need not vary monotonically with $\ell_1$ plan mismatch.

Together, the results quantify how temporal information and plan quality enter the benchmark gap, while the experiments illustrate their finite-horizon effects. Natural extensions include recency-weighted learning under nonstationary competing bids, partial threshold feedback, and model-specific guarantees that translate primitive forecast errors into $\PlanErr$.


\bibliographystyle{informs2014} 




\bibliography{literature}


  



\ifdefined\OMITONLINESUPPLEMENT
\else
\clearpage
\section*{Online Supplement}
\begin{APPENDICES}
\renewcommand{\theHsection}{appendix.\Alph{section}}
\renewcommand{\theHsubsection}{appendix.\Alph{section}.\arabic{subsection}}
\renewcommand{\theHsubsubsection}{appendix.\Alph{section}.\arabic{subsection}.\arabic{subsubsection}}
\input{major_revision/proof_3}
\input{major_revision/proof_4}
\input{major_revision/proof_5}
\end{APPENDICES}
\fi

\end{document}

%% file: major_revision/1_intro.tex
Real-time advertising requires coordinated decisions about impression value, first-price bid shading, and pacing a finite campaign budget. Valuation and bidding occur auction by auction, whereas pacing allocates the budget over time. Because opportunities change with traffic, user composition, and conversion propensity, constant-rate spending may exhaust the budget before valuable periods or preserve too much for weaker ones.

Operational systems address this challenge by combining offline planning with online control. Historical auction logs and forecasting models estimate future traffic, values, and market conditions; an offline planner converts these forecasts into a desired spending trajectory; and an online controller adjusts bidding aggressiveness when realized spending deviates from the target~\citep{lee2013real,agarwal2014budget,balseiro2024field}. Recent theory studies expenditure plans in second-price auctions and spending targets in general online allocation~\citep{balseiro2023robust,stradi2025no}, while first-price auction learning typically uses an aggregate budget or a constant spending rate. This leaves a more specific interface question: \emph{How can a forecast-generated spending trajectory guide first-price bids when private values evolve over time and the competing-bid distribution is unknown, while still allowing online adaptation and hard budget control?}

We study this question through a \emph{budget allocation plan}, a feasible vector of target expenditures. The plan serves as an operational spending schedule, the control signal in the online dual update, and the primitive defining strict and relaxed pacing comparators. Our main result connects these roles by decomposing the global expected-budget gap into plan mismatch and fixed-plan execution while preserving pathwise budget feasibility. The controller builds on empirical-CDF first-price bidding and projected dual methods~\citep{balseiro2020dual,balseiro2022best,balseiro2023horizon,wang2023learning}; our analysis quantifies how an arbitrary feasible plan affects global performance under changing private-value distributions.

\subsection{Model and Information Structure}

We consider one bidder with initial budget $B\in(0,Tb]$ participating in a horizon of $T\geq2$ sequential first-price auctions, where $[T]\triangleq\{1,\ldots,T\}$. The excluded case $B=0$ is immediate: bidding zero in every period is the only useful policy. In period $t\in[T]$, the bidder observes a private value $v_t\in[0,b]$, where $0<a<b<\infty$, and submits a bid
\[
x_t\in[0,\min\{b,B_t\}],
\]
where $B_t$ is the remaining budget at the beginning of the period and $B_1=B$. A zero bid means that the bidder skips the auction. Let $m_t\in[a,b]$ be the highest competing bid. The bidder wins when $x_t\geq m_t$, receives surplus
\[
r(x_t,v_t,m_t)\triangleq (v_t-x_t)\mathbf 1\{x_t\geq m_t\},
\]
and spends $z_t\triangleq x_t\mathbf 1\{x_t\geq m_t\}$. Thus $B_{t+1}=B_t-z_t$, and the bid cap $x_t\leq B_t$ guarantees pathwise budget feasibility.

\begin{remark}[Ranges and the opt-out action]
The competing threshold is supported on $[a,b]$, where $a>0$ is the platform reserve or smallest positive market price. Private values may lie in $[0,b]$, and zero is included in the bidder's action space. If $v_t<a$, every bid $x\geq a$ yields nonpositive surplus, while every bid $x<a$ has zero winning probability. Under the smallest-maximizer convention, the bidder therefore selects zero, incurs no expenditure, and obtains zero reward.
There is no requirement that the bidder reserve $a$ units for every future auction. Such a restriction would apply only if participation in every remaining auction were mandatory. Here the bidder may skip any auction, so pathwise feasibility requires only that the current positive bid not exceed the remaining budget $B_t$.
\end{remark}

\noindent\textbf{Competing bids.}
The highest competing bids are drawn independently from an unknown distribution with cumulative distribution function (CDF) $G$, are independent of the bidder's private values, and are revealed after each auction. The current threshold $m_t$ is unavailable before bidding, so the bidder estimates $G$ from past thresholds. We extend the CDF by $G(x)=0$ for $x<a$. Conditional on $(v_t,x_t)$, expected reward and expenditure are therefore
\[
r(x_t,v_t)\triangleq \E_{m_t}[r(x_t,v_t,m_t)] =(v_t-x_t)G(x_t),
\quad c(x_t)=x_tG(x_t).
\]

We impose no density, smoothness, or continuity condition on $G$; in particular, the competing-bid distribution may be discrete. Our model separates changes in the bidder's private values from changes in the competitive environment. For example, advertiser-specific changes in product margins or conversion values can alter the value of an opportunity without requiring a corresponding change in competing bids. We represent the former through the time-varying distributions $F_t$ and maintain a common competing-bid distribution $G$ over the campaign.

This approximation does not cover changes in market conditions that also shift competing bids, such as time-of-day effects, changes in traffic composition, or changes in other bidders' strategies. A full-history empirical CDF may then differ systematically from the current distribution $G_t$. Extending the guarantee would require controlling this discrepancy; rolling-window or discounted estimators are possible approaches, but are not analyzed here. The present theoretical and numerical results maintain a stationary $G$.

\noindent\textbf{Private values and information regimes.}
The private values are independent across periods, with $v_t\sim F_t$, and the sequence $\mathcal F=\{F_t\}_{t=1}^T$ is unknown to the bidder. The distributions $F_t$ need not be identical. This model captures systematic changes in the bidder's opportunity quality while isolating the separate statistical problem of learning the stationary competing-bid distribution $G$. We study an uninformative regime, in which no prior temporal allocation signal is available, and an informative regime, in which the bidder receives a predicted spending plan rather than the distributions $\{F_t\}_{t=1}^T$ themselves.

\noindent\textbf{Budget allocation plan.}
A plan is a vector $\bhrho=(\hrho_1,\ldots,\hrho_T)$ in the capped budget simplex
\[
\mathcal C_B\triangleq \left\{q\in[0,b]^T:\ \sum_{t=1}^Tq_t=B\right\}.
\]

The bound $B\leq Tb$ is without loss because more than $Tb$ cannot be spent. The target $\hrho_t$ is the desired expected expenditure in period $t$ and acts as a \emph{soft control signal}: realized spending may differ, and the dual variable corrects later bids. By contrast, the strict plan benchmark constrains its comparator's expected expenditure period by period.

\subsection{Constructing a Feasible Plan from Operational Forecasts}\label{subsec:plan-construction}

This subsection defines a modular plan--control interface. Before the campaign, an upstream module converts predicted operational primitives or managerial priorities into a feasible plan, which Algorithm~\ref{alg:1} uses as a soft target. The guarantee evaluates the final plan through Plan Error $\PlanErr$, defined in Section~\ref{sec:known_F}, independently of the forecasting technology; a model-specific stability relation can translate primitive errors into this metric. The constructions below are related to informative targets and sample-generated expenditure plans~\citep{jiang2025online,balseiro2023robust}.

A practical plan can be constructed through the following three stages.

\noindent\textbf{1. Forecast operational primitives.}
Using historical auction logs or campaign-level forecasting models, the advertiser first estimates period-specific traffic and private-value characteristics together with a competing-bid CDF for the campaign. Denote these planning inputs by $(\widetilde F_t,\widetilde G)$. We use a tilde to distinguish the offline planning estimate $\widetilde G$ from the online empirical CDF $\widehat G_t$ learned in Section~\ref{sec:main_alg}. These are the types of quantities that practical forecasting systems can estimate directly. 

\noindent\textbf{2. Optimize and normalize offline.}
The offline planner next converts the forecasts into a feasible period-by-period spending plan. For each period $t$ and candidate expected-expenditure cap $c\in[0,b]$, let $\widehat\Phi_t(c)$ denote the largest estimated expected surplus attainable in that period while keeping expected expenditure at most $c$:
\[
\widehat\Phi_t(c) \triangleq \sup_{\psi_t}\left\{
\E_{\widetilde F_t,\widetilde G}\left[(v-\psi_t(v))\mathbf{1}\{\psi_t(v)\geq m\}\right]:
\E_{\widetilde F_t,\widetilde G}\left[\psi_t(v)\mathbf{1}\{\psi_t(v)\geq m\}\right]\leq c\right\}.
\]
Here, $\psi_t$ is a measurable, possibly randomized bidding rule that maps the observed value to a bid in $[0,b]$, including the zero bid. Both expectations also average over any internal randomization of this rule.
The planner then allocates the total budget across periods by solving
\begin{equation}\label{eq:plan-construction}
\bhrho \in \argmax_{q\in\mathcal C_B} \sum_{t=1}^T\widehat\Phi_t(q_t).
\end{equation}

Thus, $q_t$ is the expected spending capacity assigned to period $t$, and the planner chooses the feasible allocation that maximizes the predicted total surplus. When several plans attain the same objective value, we use the minimum-Euclidean-norm solution as a deterministic tie-breaking convention. The online guarantees apply to any feasible plan, whether or not it is obtained from this optimization.

A forecasting system need not provide complete distributional forecasts. It may instead produce nonnegative opportunity scores $w_t$. When $\sum_{s=1}^T w_s>0$, these scores can be converted into a raw spending profile by
\[
\widetilde\rho_t=\frac{Bw_t}{\sum_{s=1}^{T}w_s},
\quad t\in[T].
\]

This rule assigns more budget to periods with larger predicted opportunities. If all scores are zero, we use the Uniform Plan $\widetilde\rho_t=B/T$. If the raw profile satisfies $\widetilde\rho_t\leq b$ for every $t$, it is already feasible. Otherwise, we choose the closest feasible profile by projecting it onto the capped simplex:
\[
\bhrho \in \argmin_{q\in\mathcal C_B} \sum_{t=1}^T(q_t-\widetilde\rho_t)^2.
\]

Under either construction, the final plan is nonnegative, satisfies the one-period cap, and sums to the total budget $B$. Because this final plan is the signal received by the online algorithm, plan quality is evaluated after normalization or projection.

\noindent\textbf{3. Correct online.}
The generated plan is not treated as a rigid spending requirement. During the campaign, the algorithm observes auction outcomes, updates its estimate of $G$ from revealed thresholds, solves an empirical first-price Lagrangian bidding problem, and adjusts its dual price according to the difference between planned and realized spending. If realized spending exceeds the target, future bids become more conservative; if spending falls behind the target, future bids become more aggressive. Thus, the plan acts as a soft control signal that provides temporal guidance while allowing online adaptation. Section~\ref{sec:known_F} measures the discrepancy between the final feasible plan received by the algorithm and a selected ideal allocation; this plan error is distinct from raw forecast error in a particular forecasting model.

\subsection{Benchmark and Regret}

We evaluate an online bidding policy against a benchmark with substantially more information. Let $\Pi$ denote the class of feasible online policies. A policy $\pi\in\Pi$ may use the supplied ex-ante plan, the current private value $v_t$, previous auction outcomes and thresholds, and the remaining budget. It does not know future private values or thresholds.

For $\pi\in\Pi$, its expected cumulative reward is
\[
V^\pi \triangleq
\E\left[\sum_{t=1}^{T} (v_t-x_t^\pi)G(x_t^\pi)\right].
\]

The expectation is taken over the randomness of private values, auction feedback, and any randomization in the policy. The current competing bid $m_t$ does not appear explicitly in this expression because it has already been integrated through the win probability $G(x_t^\pi)$. Recall that $m_t$ is observed only after the bid is submitted and cannot be used when making the current decision.

Our benchmark is the standard expected-budget fluid relaxation used in budget-constrained auction learning~\citep[e.g.,][]{balseiro2019learning}. Let $\Pi_0\supseteq\Pi$ denote the class of nonanticipative, distribution-informed bid rules with actions in $[0,b]$. These rules know $G$ and $\{F_t\}_{t=1}^{T}$, observe the current private value before bidding, and do not observe the current competing bid or any future realization. 
Unlike policies in $\Pi$, rules in $\Pi_0$ need not track a remaining budget on every sample path; their expected expenditure is constrained directly in the fluid problem below. Thus the benchmark differs from the online policy in both statistical information and budget feasibility.
The benchmark is defined as
\begin{equation}\label{eq:offline}
\begin{aligned}
\VOPT \triangleq \sup_{\pi\in\Pi_0} \
& \E\left[\sum_{t=1}^T(v_t-x_t^\pi)G(x_t^\pi)\right]\\
\text{s.t. } \
& \sum_{t=1}^T \E\left[x_t^\pi G(x_t^\pi)\right]\leq B, \quad x_t^\pi\in[0,b].
\end{aligned}
\end{equation}

This expected-budget relaxation is tractable and upper bounds every hard-budget policy. Specifically, let $\pi^{\mathrm{HB}}$ be an optimal nonanticipative, hard-budget policy under the same information structure as $\Pi$, and let $V^{\mathrm{HB}}(\bm\gamma)$ be its realized payoff on sample path $\bm\gamma$. Then $\E\big[V^{\mathrm{HB}}(\bm\gamma)\big]\leq\VOPT$. Thus, $\VOPT$ is a distribution-informed, nonanticipative fluid benchmark: it knows $G$ and $\{F_t\}_{t=1}^T$, observes the current private value, and never observes the current threshold or any future realization.
For an online policy $\pi$, we define regret as
\[
R_T(\pi)=\VOPT-V^\pi,
\]
which is the performance gap to the distribution-informed expected-budget fluid benchmark.

\subsection{Main Results and Contributions}
This paper establishes a plan-to-performance guarantee for budget-constrained first-price auctions. The online policy uses any feasible spending plan as a soft control signal rather than a hard period-by-period constraint. Its global expected-budget gap separates the economic loss caused by an inaccurate plan from the execution loss incurred while learning the competing-bid CDF and maintaining the budget on every sample path. The analysis connects a periodwise benchmark representation, an ideal allocation supported by a common shadow price, a plan-sensitivity bound, and hard-budget control.

The controller builds on empirical-CDF first-price bidding and projected dual control~\citep{wang2023learning,balseiro2020dual,balseiro2022best}. Variable targets, informative allocations, and spending-plan comparators are established antecedents~\citep{balseiro2023horizon,jiang2025online,balseiro2023robust,stradi2025no}. Here these ideas are developed into execution and plan-sensitivity guarantees for first-price auctions with changing private-value distributions and unknown stationary competition.

Theorem~\ref{thrm:discussion1} bounds execution under any feasible plan. Its proof compares the auxiliary policy with the period-cap benchmark at the current predictable dual price; this comparison does not require the period objectives to share a minimizer. Periodwise weak duality is combined with uniform empirical-CDF control and the projected-dual inequality. The remaining-budget check is handled by coupling the auxiliary and implemented policies and bounding the reward loss through auxiliary excess expenditure. The positive lower bound $a$ controls this correction, so no positive minimum period target is imposed. Lemma~\ref{lem:decomposition} subsequently constructs the reference allocation, including ties and slack budget, and Proposition~\ref{prop:plan-sensitivity} relates its distance from the supplied plan to global benchmark loss. Together these results yield Theorem~\ref{thrm:informative}.

\noindent\textbf{A practical plan--control interface.}
When no timing information is available, the algorithm uses the uniform plan $\hrho_t=B/T$; otherwise, an offline planner supplies a nonuniform feasible plan. In both cases the plan is a soft target, so bids may react to realized values, payments, and the learned competing-bid distribution.

\noindent\textbf{Global guarantees with and without timing information.}
Without a predicted plan, the uniform target gives a structure-free baseline. The resulting policy satisfies
\[
R_T(\pi)\leq O(\sqrt T)+O(\mathcal W_T),
\]
where $\mathcal W_T$ measures temporal heterogeneity in the private-value distributions. The first term is the cost of learning and budget control. The second is the possible cost of spending uniformly when opportunity quality changes over time.

With a predicted plan, the relevant information term is the quality of that plan. Let $\rho_t$ be period $t$ of a fixed ideal allocation selected independently of the prediction, and define
\[
\PlanErr\triangleq\sum_{t=1}^{T}|\rho_t-\hrho_t|,
\]
the $\ell_1$ distance between the final normalized plan and that ideal allocation. If $\VPLAN(\bhrho)$ is the strict benchmark whose expected spending respects the plan in every period, then
\[
\VOPT-\Vpi = \bigl(\VOPT-\VPLAN(\bhrho)\bigr) + \bigl(\VPLAN(\bhrho)-\Vpi\bigr).
\]
The first term is plan-mismatch loss and is at most $(b/2a)\PlanErr$. The second is the one-sided fixed-plan execution gap, which is bounded above by $O(\sqrt{T})$. Thus
\[
R_T(\pi)\leq O(\sqrt{T})+\frac{b}{2a}\PlanErr.
\]
The decomposition separates the contributions of upstream plan quality and online execution in the performance guarantee. The algorithm neither observes nor needs $\PlanErr$. Exact evaluation generally requires the underlying distributions and the selected ideal allocation; a realized campaign path alone does not identify $\PlanErr$.

\noindent\textbf{Pacing comparators and lower bounds.}
The strict comparator asks how a distribution-informed expected-budget benchmark would perform if it respected the same expected spending schedule period by period. A relaxed comparator allows the evaluator to add period-specific tolerances $\epsilon_t\geq0$ while retaining the global budget. Its comparator gap increases by at most $(b/a)\sum_t\epsilon_t$, and a particular choice of tolerances recovers the global benchmark exactly. Separate lower bounds show that linear dependence on $\mathcal W_T$ and on $\PlanErr$ is unavoidable in the corresponding worst-case classes. They establish the necessity of the two information terms separately.

The results allow arbitrary changes in private-value distributions while maintaining a stationary competing-bid CDF learned from revealed thresholds. They condition on the final normalized plan; translating primitive traffic, value, or competing-bid forecast errors into $\PlanErr$ requires a model-specific stability relation.

\subsection{Related Literature}

\noindent\textbf{Auction foundations and market context.}
The economics of sealed-bid auctions explains strategic bid shading, the role of information, and the design of selling mechanisms~\citep{vickrey1961counterspeculation,wilson1969communications,myerson1981optimal,riley1981optimal}. Classic treatments place these ideas in the broader theory of auctions~\citep{wilson1985game,klemperer2004auctions}. In online advertising, exchange design and the move toward first-price selling make bid shading an everyday operational problem~\citep{mcafee2011design,choi2020online,despotakis2021first}. These papers provide the economic and market-design background; our problem is the sequential decision problem of one budget-constrained bidder who must learn from auction feedback.

\noindent\textbf{Repeated first-price learning and feedback.}
What a bidder can learn depends on what the auction reveals. Prior work studies how information feedback changes first-price bidding~\citep{esponda2008information}, how bids can be learned across repeated auctions under stochastic or adversarial inputs~\citep{weed2016online,han2020learning,han2024optimal}, and how nonparametric or contextual bid rules can use richer data~\citep{zhang2021meow,badanidiyuru2023contextual}. Together, this literature spans binary-outcome, winning-bid-only, full-threshold, contextual, and adversarial feedback models. These results explain how to learn bids from auction feedback, but they do not study how an external period-level spending plan should enter hard-budget control under changing bidder values.

\noindent\textbf{Budget-constrained auction learning.}
Budgets affect bidding incentives in static and dynamic auctions~\citep{che1996expected,che1998standard,kotowski2020first,iyer2014mean,balseiro2015repeated,balseiro2019learning}. \citet{ai2022no} study budgeted first-price learning using an RL-based approach. \citet{wang2023learning} provide the algorithmically closest empirical-CDF and dual-pacing controller, with a constant average-consumption target under i.i.d. values. Our analysis treats an arbitrary feasible external plan under changing private-value distributions and connects fixed-plan execution to global-to-plan sensitivity and strict and relaxed period-cap comparators. Contextual budgets, one-sided feedback, and newer first-price autobidding provide complementary information and constraint models~\citep{balseiro2023contextual,fu2026contextual,li2025autobidding}.

\noindent\textbf{Online allocation, knapsack constraints, and nonstationarity.}
Bid-price control has a long history in revenue management~\citep{talluri1998analysis}. More generally, dual mirror descent methods update resource prices from target consumption and realized usage~\citep{balseiro2020dual,balseiro2022best}, with variable target-consumption sequences also studied under horizon uncertainty~\citep{balseiro2023horizon}. Related resource-control ideas appear in bandits with knapsacks under stochastic and adversarial feedback~\citep{badanidiyuru2018bandits,immorlica2022adversarial,castiglioni2022online}. Closest to our setting, \citet{jiang2025online} construct informative targets supported by a common optimal dual price and allow stochastic realized rewards and consumptions. Their action optimization uses conditional-mean reward and consumption functions; in our setting, these depend on the unknown competing-bid distribution $G$, so bids instead use empirical-CDF estimates from previously revealed thresholds. We also use auxiliary-process budget comparisons and reward-to-consumption bounds from this online-allocation framework.
Nonstationary optimization and knapsack models describe changing environments through variation measures or time-varying-demand advice~\citep{besbes2014stochastic,besbes2015non,liu2022non,lyu2023bandits}. Following the Wasserstein perspective of \citet{jiang2025online}, we measure variation in the bidder's private-value distributions. \citet{liu2022non} distinguish local changes in expected rewards and consumptions from global deviations around their horizon averages. Our $\mathcal W_T$ is likewise global, but compares period-specific private-value distributions with their campaign-wide mixture. Under stationary $G$, inequality~\eqref{eq:dual-W-lipschitz} links this distance directly to changes in the optimized period Lagrangian. A complementary approach in \citet{fikioris2023approximately} uses ratios between each action's minimum and maximum expected reward and consumption over time, yielding guarantees between stochastic and adversarial regimes. These ratio-based conditions capture relative changes in conditional means, whereas $\mathcal W_T$ measures additive distributional deviations in private values.

\noindent\textbf{Predictions, expenditure plans, and pacing comparators.}
Operational advertising systems use response and cost forecasts to support delivery and pacing decisions~\citep{chapelle2014simple,oldenburg2024interpretable,lee2013real,agarwal2014budget}. A broader pacing literature explains how budget and return-on-spend constraints shape practical controllers~\citep{balseiro2024field}. The closest plan-construction antecedent is \citet{balseiro2023robust}, who build a sample-generated expenditure plan for nonstationary second-price auctions. Their robustness guarantee analyzes a specified single-sample pacing procedure, whereas ours evaluates the final feasible plan in the first-price model. Relating discrepancies between sample and true distributions to $\PlanErr$ requires a stability bound for the chosen planner. \citet{stradi2025no} study spending-plan comparators and relaxed period allowances that retain the aggregate budget. Their online-allocation model observes reward and cost functions before action selection, whereas their post-decision-feedback model compares against a fixed, possibly randomized strategy. Our comparator allows period-specific bid rules contingent on the observed private value, while the current threshold is revealed only after bidding. \citet{lucier2024autobidders} use a benchmark of per-round value-optimizing pacing multipliers subject to time-averaged budget and ROI constraints applied in each round. Our strict comparator instead maximizes expected surplus over period-specific, value-contingent bid rules under the supplied caps $\hrho_t$, with only a budget constraint. Thus, period-constrained comparators are an antecedent; our focus is their connection to fixed-plan execution and global plan sensitivity in the present first-price model. Budget- and ROI-constrained regret guarantees are also studied by \citet{aggarwal2025noregret}.

%% file: major_revision/2_dual_problem_main_algorithm.tex
This section formulates the controller analyzed in Sections~\ref{sec:unknown_F}--\ref{discussion}. A feasible spending plan enters the projected dual update as a time-varying soft target, the empirical-CDF bid rule learns the unknown competing-bid environment, and the remaining-budget check enforces pathwise feasibility. This construction makes the interface between offline planning and online first-price control explicit.

\subsection{Lagrangian Dual Formulation}

Recall that $\VOPT$ is the fluid expected-budget benchmark introduced in Section~\ref{sec:intro}. This benchmark knows the statistical environment and is constrained only in expected total spending. To account for that constraint, we attach a nonnegative multiplier $\mu\geq0$ to each dollar of expected expenditure. This multiplier is the shadow price of the remaining campaign budget. Relaxing the expected-budget constraint gives the Lagrangian objective
\begin{equation}\label{eq:offlineLR}
\VLR(\mu) =\mu B+\sum_{t\in[T]}\E_{v_t}\left[\max_{x\in[0,b]}\big(v_t-(1+\mu)x\big)G(x)\right].
\end{equation}

The maximization term describes one auction under a fixed budget price. Without a budget price, bidding $x$ produces expected surplus $(v-x)G(x)$. The multiplier adds the opportunity cost $\mu xG(x)$ of using scarce budget now, so it raises the effective cost of a winning bid from $x$ to $(1+\mu)x$. A larger $\mu$ therefore makes the bidder more conservative.

The option $x=0$ means skipping the auction, and $G(x)=0$ for $x<a$ because competing bids are supported on $[a,b]$. Throughout the paper, every bid maximization uses the same deterministic convention: among all maximizers, we select the smallest bid. In particular, whenever zero is optimal, the selected bid is zero. Under the maintained assumptions, the single-period objective attains its maximum on $[0,b]$. Choosing the smallest maximizer therefore defines a unique bid as a measurable function of the value and the dual variable. After the expected-budget constraint is moved into the objective, every policy feasible for the fluid benchmark is also feasible for the Lagrangian relaxation, which gives the standard weak-duality relationship.

\begin{lemma}[Weak Duality]\label{lemma:weak-duality}
For every $\mu\geq0$, $\VOPT\leq\VLR(\mu)$.
\end{lemma}

To express the aggregate budget period by period, choose any accounting vector $\mathbf q=(q_1,\ldots,q_T)\in\mathcal C_B$. Thus, $0\leq q_t\leq b$ and $\sum_{t\in[T]}q_t=B$. At this point, $\mathbf q$ is only a bookkeeping device: replacing $\mu B$ by $\sum_t\mu q_t$ does not change the Lagrangian or impose period-level spending constraints.
Using this allocation, the Lagrangian objective can be written as
\begin{equation}\label{eq:offlineLR2}
\begin{aligned}
\VLR(\mu) & =\sum_{t\in[T]}\left\{\mu q_t+\E_{v_t}\left[\max_{x\in[0,b]}\big(v_t-(1+\mu)x\big)G(x)\right]\right\}\\
&\triangleq\sum_{t\in[T]}D_t^{\mathbf q}(\mu).
\end{aligned}
\end{equation}

The neutral symbol $\mathbf q$ keeps this decomposition separate from the two spending profiles introduced later. The online controller receives the predicted plan $\bhrho$, whereas Section~\ref{sec:known_F} constructs the prediction-independent ideal allocation $\brho^*$. This distinction prevents a bookkeeping identity from being mistaken for either an input to the algorithm or an optimal allocation.

For later use, define
\[
x^*(v,\mu) \in \argmax_{x\in[0,b]}\big(v-(1+\mu)x\big)G(x).
\]

The corresponding subgradient of $D_t^{\mathbf q}(\mu)$ is $q_t-\E_{v_t}\left[x^*(v_t,\mu)G\big(x^*(v_t,\mu)\big)\right]$.
The first term is the accounting target, and the second is the expected spending induced by the current budget price. Their difference has the same form as the feedback signal used by the online controller below.

We write $\mu^*\in\argmin_{\mu\geq0}\VLR(\mu)$. The Lagrangian construction and projected-dual update are inherited from online resource allocation and budget-pacing methods. The new analysis studies what happens when a forecast-generated spending plan is used inside this template: it separates the loss caused by plan mismatch from the loss caused by online execution, learning $G$, and enforcing a hard budget.

\subsection{Plan-Driven Dual Update}\label{subsec:plan-dual-update}

Algorithm~\ref{alg:1} implements the plan-driven bidding policy. The key idea is to use the allocation plan as a temporal spending target rather than a hard constraint. In each period, the bidder solves a single-auction bidding problem using the current estimate of the competing-bid distribution and then adjusts a dual price according to the difference between planned and realized spending. This dual price controls how aggressively the bidder participates in future auctions.

The algorithm receives a feasible plan $\bhrho=(\hrho_1,\ldots,\hrho_T)$ where $0\leq \hrho_t\leq b, \ \sum_{t=1}^{T}\hrho_t=B$.
The construction of such a plan is discussed in Section~\ref{subsec:plan-construction}. The per-period upper bound $b$ is natural because one auction cannot require more than the maximum possible bid. If an upstream forecasting system produces a larger target, the capped normalization procedure described earlier converts it into a feasible plan.

At the beginning of period $t$, the bidder has observed the competing bids from previous auctions and uses them to estimate the unknown competing-bid CDF. For $t\geq2$, the empirical estimate is
\begin{equation}\label{eq:empirical-G}
\widehat G_t(x)=\frac{1}{t-1} \sum_{s=1}^{t-1}\mathbf 1\{m_s\leq x\},
\quad x\in[0,b].
\end{equation}

Because the current competing bid $m_t$ is revealed only after submitting the bid, it cannot be used in the current decision. For the first auction, we use the initialization $\widehat G_1(x)=\mathbf 1\{x\geq a\}$, which satisfies $\widehat G_1(0)=0$. This initialization is used only before threshold observations are available. Its first-period uniform CDF error is at most one, so it contributes only a constant to the cumulative estimation-error bound.
Given the current budget price $\mu_t$, the bidder chooses the bid that maximizes empirical adjusted surplus, using the smallest-maximizer convention stated above:
\[
\tilde x_t\in \argmax_{x\in[0,b]}(v_t-(1+\mu_t)x)\widehat G_t(x).
\]

Here $1+\mu_t$ is the effective cost per dollar paid in the current auction: one dollar is the payment itself, and $\mu_t$ is the estimated value of preserving that dollar for later. The remaining-budget check guarantees hard sample-path feasibility. If the desired bid exceeds the money left, the bidder skips the auction.

After the auction, the bidder observes the realized competing bid $m_t$, obtains the realized expenditure $z_t=x_t^\pi\mathbf 1\{x_t^\pi\geq m_t\}$, and updates the dual price using the difference between planned and realized spending: $g_t=\hrho_t-z_t,\ \mu_{t+1}=(\mu_t-\eta g_t)^+$.

\begin{algorithm}
\SetAlgoLined
\SetKwInput{KwInput}{Input}
\SetKwInput{KwOutput}{Output}
\KwInput{Horizon $T$; budget $B$; step size $\eta>0$; feasible plan $\bhrho$}
\KwOutput{A hard-budget-feasible bidding policy $\pi$}
Initialize $\mu_1\geq0$ and $B_1=B$\;
\For{$t=1,\ldots,T$}{
Observe $v_t\in[0,b]$ and construct $\widehat G_t$ from $m_1,\ldots,m_{t-1}$\;
Set $\tilde x_t$ to the smallest maximizer of $(v_t-(1+\mu_t)x)\widehat G_t(x)$ over $x\in[0,b]$\;
Bid $x_t^\pi=\tilde x_t$ if $\tilde x_t\leq B_t$, and bid $x_t^\pi=0$ otherwise\;
Observe $m_t$ and set $z_t=x_t^\pi\mathbf 1\{x_t^\pi\geq m_t\}$\;
Update $g_t=\hrho_t-z_t$ and $\mu_{t+1}=(\mu_t-\eta g_t)^+$\;
Update $B_{t+1}=B_t-z_t$\;
}
\caption{Plan-driven adaptive bidding policy}\label{alg:1}
\end{algorithm}

\subsection{Operational Interpretation}

Before the campaign, an offline planner converts forecasts or managerial priorities into the feasible target schedule $\bhrho$. The initial dual price may be set by a separate ex-ante calibration model; the analysis permits any $0\leq\mu_1\leq b/a+b$. Section~\ref{numerical} uses one common timing-agnostic calibration for all plan-driven controllers. In each auction, the controller observes $v_t$, updates an empirical estimate of the competing-bid CDF $G$, chooses a bid, and then compares the realized payment with that period's target.

The sign of the correction has a direct pacing interpretation. If realized spending exceeds $\hrho_t$, then $g_t<0$ and $\mu_{t+1}$ increases, making subsequent bids more conservative. If realized spending falls short, then $g_t>0$ and the dual price decreases, making subsequent bids more aggressive. Thus the plan is a soft control signal rather than a hard per-period constraint on the online policy. The remaining-budget check nevertheless guarantees sample-path budget feasibility.

The update increment $g_t$ has the target-minus-expenditure form of a dual subgradient, but it also reflects the supplied plan, empirical-CDF approximation, and the remaining-budget check. The analysis controls these discrepancies rather than assuming that the implemented update is an unbiased subgradient. The results below bound cumulative reward loss while preserving hard budget feasibility; they do not require pointwise convergence of the dual iterates.

%% file: major_revision/3_uninformative_case.tex
We first consider the uninformative regime. The bidder does not know the sequence of private-value distributions $\{F_t\}_{t\in[T]}$ and receives no forecast-generated signal about when high-value opportunities are likely to arrive. Algorithm~\ref{alg:1} therefore uses the feasible uniform plan $\hrho_t=B/T$ in every period. This plan is a soft target, not a requirement to spend exactly $B/T$ each time. It provides a structure-free default, but it cannot direct spending toward periods with better opportunities. We now measure the possible cost of that missing timing information.

\subsection{A Wasserstein Measure of Temporal Heterogeneity}\label{subsec:Wasserstein-dist}

\noindent\textbf{What does $\mathcal W_T$ measure?}
Because the private-value distribution may differ across periods, we first need a distance between two value distributions. For distributions $F$ and $F'$ on $[0,b]$, let $\Gamma(F,F')$ denote the set of their couplings. A coupling is a joint distribution of two random values whose marginal distributions are $F$ and $F'$, respectively. The one-Wasserstein distance is
\begin{equation}\label{eq:Wasserstein_F1_F2}
\mathcal W_1(F,F') \triangleq \inf_{\gamma\in\Gamma(F,F')} \int_{[0,b]^2}|v-v'|\,\mathrm d\gamma(v,v').
\end{equation}

Equivalently, $\mathcal W_1(F,F')$ is the smallest possible expected distance between a draw from $F$ and a draw from $F'$ over all ways of pairing the two draws.

This metric captures the \emph{magnitude} of a change in values. For example, if $F$ places all its mass at $v$ and $F'$ places all its mass at $v+\delta$, then $\mathcal W_1(F,F')=|\delta|$. Thus, two distributions with nearby but nonoverlapping supports are regarded as close when the underlying shift is small. This feature makes the Wasserstein distance more informative for gradual changes in conversion values or user composition than divergence measures that treat disjoint supports as completely separated. Its use as a nonstationarity measure follows the Wasserstein-based framework of \citet{jiang2025online}; our analysis specializes this idea to budget-constrained first-price bidding while also accounting for online estimation of the unknown competing-bid distribution $G$.

To obtain a campaign-wide reference distribution, we define $\bar F_T \triangleq \frac{1}{T}\sum_{t\in[T]}F_t$. The distribution $\bar F_T$ can be interpreted as first selecting one of the $T$ periods uniformly at random and then drawing a private value from that period's distribution. We measure total heterogeneity of the campaign by
\begin{equation}\label{eq:WT}
\mathcal W_T \triangleq \sum_{t\in[T]}\mathcal W_1(F_t,\bar F_T).
\end{equation}

Hence, $\mathcal W_T$ measures how far the period-specific value distributions are, in total, from their campaign-wide mixture. We have $\mathcal W_T=0$ when $F_t=F$ for every period, while $0\leq\mathcal W_T\leq Tb$ in general. Small shifts in user composition or conversion value can produce a small or sublinear $\mathcal W_T$, whereas a persistent difference between period types or a large structural change can make $\mathcal W_T$ proportional to $T$.

The quantity $\mathcal W_T$ is not a consecutive-period path-length measure: it does not measure how rapidly $F_t$ changes from one period to the next. Rather, it measures period-to-campaign heterogeneity. This is the relevant comparison in the uninformative regime because the uniform plan applies one common spending target across the entire horizon and does not distinguish among periods. Notice also that $\mathcal W_T$ concerns only variation in the bidder's private-value distributions. The competing-bid distribution $G$ remains stationary, though unknown, under the theoretical model.

\noindent\textbf{Why does $\mathcal W_T$ enter the guarantee?}
Define
\[
\phi(v,\mu) \triangleq \max_{x\in[0,b]} \bigl(v-(1+\mu)x\bigr)G(x).
\]

The quantity $\phi(v,\mu)$ is the best adjusted expected surplus available to a bidder with private value $v$ when the current budget price is $\mu$. For any fixed bid $x$, the function $v\mapsto \bigl(v-(1+\mu)x\bigr)G(x)$ is affine in $v$ with slope $G(x)\in[0,1]$. Taking the pointwise maximum over $x$ therefore preserves the one-Lipschitz property. Consequently, for every $\mu\geq0$,
\begin{equation}\label{eq:dual-W-lipschitz}
\left|\E_{F}[\phi(v,\mu)]-\E_{F'}[\phi(v,\mu)]\right| \leq \mathcal W_1(F,F').
\end{equation}

Inequality~\eqref{eq:dual-W-lipschitz} has a direct interpretation: replacing one private-value distribution by another changes the expected Lagrangian-adjusted surplus by no more than their Wasserstein distance. In the proof, this permits each period-specific distribution $F_t$ to be compared with the common reference $\bar F_T$ at cost $\mathcal W_1(F_t,\bar F_T)$. Summing these comparison costs over the horizon produces $\mathcal W_T$. The argument requires no density or smoothness assumption on $G$ and remains valid when $G$ is discrete.

\subsection{Regret Guarantee}\label{subsec:performance-nonknown-case}

We now quantify the performance cost of using the uniform plan when no prior information about the timing of valuable opportunities is available. Because $B\leq Tb$, the uniform plan belongs to $\mathcal C_B$. 
\begin{thrm}\label{thrm:noninformative} 
Suppose $B/T\leq b$. Run Algorithm~\ref{alg:1} with $\hrho_t=B/T$ for every $t\in[T]$, step size $\eta=T^{-1/2}$, and $0\leq\mu_1\leq b/a+b$. Then 
\[ 
\VOPT-\Vpi \leq O(\sqrt T)+O(\mathcal W_T). 
\] 
The hidden constants depend only on $a$ and $b$.
\end{thrm} 

The bound separates two sources of performance loss. The first term is the online learning and control cost. It includes estimating the unknown competing-bid distribution $G$ from previously revealed competing bids, updating the dual price, and enforcing the hard campaign budget. The second term is the cost of using the same spending target $B/T$ in every period when the private-value distributions vary over time. Thus, $\mathcal W_T$ measures the temporal information that is missing from the uniform plan. 

\noindent\textbf{When is the guarantee sublinear?}
Algorithm~\ref{alg:1} does not need to know $\mathcal W_T$ and uses the same policy for every degree of heterogeneity. If $\mathcal W_T=o(T)$, then the regret is sublinear in $T$, so the average performance gap $(\VOPT-\Vpi)/T$ converges to zero. In contrast, when $\mathcal W_T=\Theta(T)$, the value distributions differ substantially across the horizon. The fluid benchmark can direct expected spending toward more valuable periods, whereas the uniform plan contains no timing information. In this case, sublinear regret cannot be guaranteed in general.

Here is a concrete scale example. Let $F_t=\delta_{u+\delta_t}$, where $|\delta_t|\leq cT^{-1/2}$ and every $u+\delta_t$ lies in $[a,b]$. Then
\[
\mathcal W_1(F_t,\bar F_T)
=\frac{1}{T}\sum_{s=1}^T|\delta_t-\delta_s|
\leq 2cT^{-1/2},
\]
so $\mathcal W_T\leq2c\sqrt T$ and both terms in Theorem~\ref{thrm:noninformative} are sublinear. By comparison, if the early- and late-period point masses remain separated by a fixed positive amount, each of a constant fraction of periods stays a constant distance from the campaign mixture, giving $\mathcal W_T=\Theta(T)$.

The proof proceeds in four steps. First, we introduce an auxiliary process used only for analysis that always submits its target bid, even when that bid exceeds the current remaining budget. Second, projected online-gradient arguments control its dual updates, while concentration of the empirical CDF controls the error from learning $G$. Third, inequality~\eqref{eq:dual-W-lipschitz} compares each period-specific distribution $F_t$ with the campaign-wide mixture $\bar F_T$, producing the term $O(\mathcal W_T)$. Finally, a coupling and budget-depletion argument bounds the difference between the auxiliary process and the implemented hard-budget-feasible policy. The complete proof appears in Online Appendix~\ref{app:proof-uninformative}.

The stationary private-value setting is an immediate special case. 
\begin{coro}\label{coro:iid}
If \(F_t=F\) for every \(t\), then \(\mathcal W_T=0\) and the policy in \Cref{thrm:noninformative} satisfies
\[
\VOPT-\Vpi\leq O(\sqrt{T}).
\]
\end{coro}

Although the private-value distributions are stationary in Corollary~\ref{coro:iid}, the competing-bid distribution $G$ remains unknown and must still be learned from auction feedback. This explains why the statistical learning term remains even when $\mathcal W_T=0$. The linear dependence on $\mathcal W_T$ is not merely an artifact of the analysis. Consider two environments that are identical during an initial block of periods but later make the remaining opportunities either more or less valuable. Before observing which environment is present, the bidder must decide how much budget to spend in the initial block. Spending early is costly when later opportunities improve, while saving the budget is costly when later opportunities deteriorate. The following proposition formalizes this tradeoff.

\begin{prop}[Necessity of linear temporal-heterogeneity dependence]
\label{prop:lbd-W}
Fix $0<a<b$. For every horizon $T\geq2$ and every $0\leq w\leq T(b-a)$, there are two admissible instances, both satisfying $\mathcal W_T\leq w$, such that every online policy $\pi$ incurs regret at least
\[
\frac{w}{12}
\]
on one of the two instances. Consequently, no uniform regret bound can replace the linear dependence on $\mathcal W_T$ by $o(\mathcal W_T)$.
\end{prop}

A separate $\Omega(\sqrt T)$ statistical lower bound holds even when the competing threshold is revealed after every auction in stationary repeated first-price auctions with an unknown competing-bid distribution \citep{han2024optimal}; \citet{wang2023learning} invoke the same lower bound for their budgeted full-threshold-feedback setting. Since both the stationary construction and the instances in Proposition~\ref{prop:lbd-W} belong to the class $\{\mathcal W_T\leq w\}$, the worst-case regret over admissible instances in this class, including their budgets, is $\Omega(\sqrt T+w)$. The two terms may be witnessed by different instances.

%% file: major_revision/4_informative_case.tex
Section~\ref{sec:unknown_F} studies the case in which no temporal information is available and Algorithm~\ref{alg:1} therefore uses the uniform plan. We now consider the informative regime, in which the bidder receives a feasible predicted plan $\bhrho=(\hrho_1,\ldots,\hrho_T)$ satisfying $0\leq \hrho_t\leq b$ and $\sum_{t=1}^T\hrho_t=B$. As explained in Section~\ref{sec:main_alg}, this plan is a soft target for the dual update rather than a hard period-by-period restriction on the online policy. To evaluate this information, we separate the economic cost of an inaccurate plan from the online cost of learning and controlling the budget while using that fixed plan:
\[
\text{global benchmark gap}=\text{plan mismatch}+\text{fixed-plan execution gap}.
\]

\subsection{A Strict Plan Comparator and the Fixed-Plan Execution Gap}

Recall that $\Pi_0$ is the class of distribution-informed, nonanticipative bid rules used to define the global fluid benchmark $\VOPT$ in \eqref{eq:offline}. Their feasibility is imposed through expected-expenditure constraints rather than a pathwise remaining-budget constraint. The strict plan comparator uses the same policy class and the same expected-reward objective, but replaces the single aggregate expected-budget constraint by one expected-expenditure cap for each period. In words, it asks how well a distribution-informed expected-budget benchmark could perform if it had to respect the same expected spending schedule period by period:
\begin{equation}\label{eq:optplan}
\begin{aligned}
\VPLAN(\bhrho) \triangleq \sup_{\pi\in\Pi_0} \
&\E\left[\sum_{t=1}^T(v_t-x_t^\pi)G(x_t^\pi)\right]\\
\text{s.t.} \
&\E\left[x_t^\pi G(x_t^\pi)\right]\leq\hrho_t,\quad t\in[T].
\end{aligned}
\end{equation}

Since $\sum_{t=1}^T\hrho_t=B$, every policy feasible for $\VPLAN(\bhrho)$ is also feasible for $\VOPT$, and hence $\VPLAN(\bhrho)\leq\VOPT$.
It is important to distinguish this comparator from the online policy. The constraints in \eqref{eq:optplan} apply to the comparator's expected expenditure. Algorithm~\ref{alg:1}, by contrast, uses $\hrho_t$ only as the target in its dual update and enforces the campaign budget on every sample path; its realized spending may be above or below the period target. 
The comparator and the implemented policy therefore use different feasibility notions. The quantity $\VPLAN(\bhrho)-\Vpi$ is a signed, one-sided execution gap: a negative value means that the online policy earns more than the stricter period-cap comparator.

\begin{thrm}[Fixed-plan execution gap]\label{thrm:discussion1}
Consider Algorithm~\ref{alg:1} with a feasible plan $\bhrho$, step size $\eta=T^{-1/2}$, and initial dual variable $0\leq\mu_1\leq b/a+b$. Its expected reward satisfies
\begin{equation}\label{eq:strict-plan-learning}
\VPLAN(\bhrho)-\Vpi\leq O(\sqrt{T}),
\end{equation}
where the hidden constant depends only on $a$ and $b$.
\end{thrm}

Theorem~\ref{thrm:discussion1} isolates the fixed-plan execution gap. Once $\bhrho$ is fixed, this signed gap is at most the learning and hard-budget execution term, regardless of whether the plan is close to the best global allocation. The proof follows the same main steps as the analysis in Section~\ref{sec:unknown_F}. Conditional on the past, the current dual price is fixed before $v_t$ is observed, so weak duality can be applied separately to the period-$t$ expected-expenditure cap in \eqref{eq:optplan}. The difference between the empirical CDF $\widehat G_t$ and the true CDF $G$ is controlled by the same concentration argument, and the projected dual update controls the cumulative planned-minus-realized spending term. Finally, the same auxiliary-process coupling handles the remaining-budget check and preserves hard budget feasibility. Unlike the alignment result below, this proof does not require the period-specific comparator problems to share one optimal shadow price. The complete argument is given in Online Appendix~\ref{app:proof-informative}.

\subsection{The Ideal Allocation and Plan Quality}

Theorem~\ref{thrm:discussion1} controls online performance relative to any fixed plan. To compare the online policy with the global benchmark $\VOPT$, we next construct a reference allocation that describes the period-level spending capacities supporting the global fluid benchmark. This ideal allocation is used only in the analysis; Algorithm~\ref{alg:1} neither observes nor computes it.

Online Appendix~\ref{app:proof-informative} proves existence of an equivalent compact periodwise formulation. Independently of the predicted plan, fix the minimum-Euclidean-norm optimal reward--expenditure vector and the smallest optimal budget price. Let $c_t^*$ denote the selected solution's expected expenditure in period $t$, and let $\mu^*\in\argmin_{\mu\geq0}\VLR(\mu)$ be the selected budget price. This deterministic convention makes the reference allocation and $\PlanErr$ well defined without affecting any bound. Put
\begin{equation*}
 C^* \triangleq \sum_{t=1}^T c_t^*\leq B.
\end{equation*}

When the global budget binds, $C^*=B$, so $(c_1^*,\ldots,c_T^*)$ already forms a feasible spending plan. The slack case needs one extra step: the fluid optimum may not use all available budget, whereas every input plan is normalized to sum to $B$. We therefore add unused capacity---not forced spending---to obtain a comparable feasible vector. Specifically, we distribute the amount $B-C^*$ across the remaining per-period capacities $b-c_t^*$ and define
\begin{equation}\label{eq:slack-completion}
s_t\triangleq
\begin{cases}
    0,&C^*=B,\\[2pt]
    \displaystyle\frac{B-C^*}{Tb-C^*}(b-c_t^*),&C^*<B.
\end{cases}
\end{equation}

In the slack case, $B-C^*$ is the unused campaign budget and $Tb-C^*$ is the total unused one-period capacity. Since $B\leq Tb$, the fraction in \eqref{eq:slack-completion} lies in $[0,1]$. It follows that $0\leq s_t\leq b-c_t^*$ and $\sum_{t=1}^T s_t=B-C^*$.
We define the selected ideal allocation by
\begin{equation}\label{def:rho1}
\rho_t\triangleq c_t^*+s_t, \quad t\in[T].
\end{equation}

Thus, $0\leq\rho_t\leq b$ and $\sum_{t=1}^T\rho_t=B$. If the global budget binds, then $s_t=0$ and $\rho_t=c_t^*$ is the expected period-$t$ expenditure of the selected globally optimal benchmark solution. When the Lagrangian bidding problem has ties, the selected primal expenditure need not equal that induced by the smallest-maximizer bid rule.

If the global budget is slack, the added amount $s_t$ is spending capacity, not expenditure that the optimal benchmark must actually use. Since $c_t^*\leq\rho_t$, the global optimum remains feasible under these period capacities. The selection and construction are fixed independently of the predicted plan $\bhrho$; other prediction-independent ways of distributing unused budget would lead to the same guarantee.

Let $\brho^*=(\rho_1,\ldots,\rho_T)$ and define $D_t^*(\mu)\triangleq D_t^{\brho^*}(\mu)$ using the period dual objective in \eqref{eq:offlineLR2}. The ideal allocation is constructed so that the globally optimal shadow price $\mu^*$ is also optimal for every period separately. Thus, the period-level spending capacities are aligned with one common budget price. 

\begin{lemma}[Common-shadow-price alignment]\label{lem:decomposition}
For the ideal allocation in \eqref{def:rho1},
\begin{equation}\label{eqn:decomposition}
\mu^*\in\argmin_{\mu\geq0}D_t^*(\mu), \quad t\in[T].
\end{equation}
Moreover,
\begin{equation}\label{eqn:Upperdecompose}
\VOPT=\min_{\mu\geq0}\VLR(\mu)=\sum_{t=1}^T\min_{\mu_t\geq0}D_t^*(\mu_t).
\end{equation}
\end{lemma}

The first statement says that one common shadow price is compatible with all periods under the ideal allocation. The second says that the sum of the period-level minima exactly recovers the global fluid benchmark. When the budget binds, the ideal target in each period equals the optimal expected expenditure. When the budget is slack, the optimal price is zero, and adding unused spending capacity preserves zero as a minimizer of every period objective. This alignment property identifies the prediction-independent temporal reference needed for the global-to-plan sensitivity argument. The result remains valid when $G$ is discrete or when the period bidding problem has multiple optimal bids.

We now measure the accuracy of the predicted plan. We evaluate the final feasible plan passed to Algorithm~\ref{alg:1}, after any normalization or projection, by its $\ell_1$ distance from the selected ideal allocation:
\begin{equation}\label{eqn:DeviationRho}
\PlanErr \triangleq \sum_{t=1}^T|\rho_t-\hrho_t|.
\end{equation}
The quantity $\PlanErr$ is an ex-post analytical measure: the algorithm does not observe it and does not need it as an input. It should also not be confused with $\mathcal W_T$ from Section~\ref{sec:unknown_F}. The quantity $\mathcal W_T$ measures how much the private-value distributions vary around their campaign-wide mixture, whereas $\PlanErr$ measures the distance between the predicted spending plan and the selected ideal allocation. A highly nonstationary campaign can still have small $\PlanErr$ when its forecast is accurate. Without assumptions linking the forecasting procedure to the environment, neither quantity controls the other, so the uninformative and informative guarantees are complementary rather than direct corollaries of one another.

\begin{prop}[Global-to-plan sensitivity]\label{prop:plan-sensitivity}
For every feasible predicted plan $\bhrho$,
\begin{equation}\label{eq:global-plan-sensitivity}
0 \leq \VOPT-\VPLAN(\bhrho) \leq \frac{b}{a}\sum_{t=1}^T(\rho_t-\hrho_t)^+ = \frac{b}{2a}\PlanErr.
\end{equation}
\end{prop}

The sensitivity bound in Proposition~\ref{prop:plan-sensitivity} has a direct interpretation. The globally optimal benchmark needs to be modified only in a period where its expected expenditure $c_t^*$ exceeds the predicted cap $\hrho_t$. Since $c_t^*\leq\rho_t$, the required reduction is bounded by $(\rho_t-\hrho_t)^+$. In each affected period, the proof scales the selected reward--expenditure pair toward the zero-bid pair so that expected expenditure is reduced to the predicted cap; closure points are handled by approximating randomized rules. Any positive winning payment is at least $a$, while one-period surplus is at most $b$. Therefore, the reward lost by this adjustment is at most $b/a$ times the expenditure removed. Finally, both $\brho^*$ and $\bhrho$ sum to $B$, so $\sum_t(\rho_t-\hrho_t)^+=\PlanErr/2$. No density, smoothness, or continuity assumption on $F_t$ or $G$ is required.

Combining the economic loss caused by plan error with the learning and control loss for a fixed plan yields the main guarantee for the informative regime.

\begin{thrm}[Plan-to-global performance guarantee]\label{thrm:informative}
Under the conditions of Theorem~\ref{thrm:discussion1},
\begin{equation}\label{eq:informative-decomposition}
\begin{aligned}
\VOPT-\Vpi
& = \bigl(\VOPT-\VPLAN(\bhrho)\bigr) + \bigl(\VPLAN(\bhrho)-\Vpi\bigr)\\
& \leq O(\sqrt{T})+\frac{b}{2a}\PlanErr.
\end{aligned}
\end{equation}
\end{thrm}

The decomposition in Theorem~\ref{thrm:informative} separates two different questions. The first term measures the economic cost of using an imperfect spending plan. The second is the fixed-plan execution gap: it measures the cost of learning the unknown competing-bid distribution and controlling the budget online while using that plan. When the supplied plan equals the selected ideal allocation, $\PlanErr=0$ and only the $O(\sqrt{T})$ learning and hard-budget execution term remains. More generally, if $\PlanErr=o(T)$, then the average global performance gap converges to zero. Algorithm~\ref{alg:1} does not need to know $\PlanErr$.

The linear dependence on $\PlanErr$ is also necessary. To see the basic difficulty, consider two environments that provide the same predicted plan and are indistinguishable during an initial block of periods, but later place the more valuable opportunities in different blocks. A policy that spends early performs poorly in one environment, while a policy that saves its budget performs poorly in the other. The following proposition formalizes this argument.

\begin{prop}[Necessity of linear Plan Error dependence]\label{prop:lbd-V}
Fix $0<a<b$. For every integer $1\leq k\leq\lfloor T/2\rfloor$, there are two instances with the same feasible predicted plan and with $\PlanErr=ak$ in either instance such that every nonanticipative online policy incurs regret at least
\begin{equation*}
\frac{b-a}{8a}\PlanErr
\end{equation*}
on one of the two instances. Consequently, no uniform upper bound can replace the linear dependence on $\PlanErr$ by $o(\PlanErr)$.
\end{prop}

A separate $\Omega(\sqrt T)$ statistical lower bound already holds under full-threshold feedback in stationary repeated first-price auctions with an unknown competing-bid distribution~\citep{han2024optimal}; see also \citet{wang2023learning} for its use in the budgeted first-price setting. Thus, the statistical learning term and the linear Plan Error term in Theorem~\ref{thrm:informative} are individually necessary. The two constructions establish separate necessities and need not be witnessed by the same instance.

%% file: major_revision/5_discussion.tex
Section~\ref{sec:known_F} introduced the strict plan comparator, which requires the distribution-informed expected-budget benchmark to respect the predicted expected-expenditure cap in every period. We now allow limited flexibility around these period caps while retaining the same global expected budget. The resulting allowances modify only the comparator used for evaluation; they do not change Algorithm~\ref{alg:1} or its dual update.

\subsection{Allowing Period-Level Flexibility}

Let $\bepi=(\epsilon_1,\ldots,\epsilon_T)$ be a nonnegative allowance vector. Using the same policy class $\Pi_0$ and the same expected-reward objective as in Section~\ref{sec:known_F}, define
\begin{equation}\label{eq:optplan_error}
\begin{aligned}
\VPLAN(\bhrho,\bepi)\triangleq \sup_{\pi\in\Pi_0} \
&\E\left[\sum_{t=1}^T(v_t-x_t^\pi)G(x_t^\pi)\right]\\
\text{s.t.} \
&\E\left[x_t^\pi G(x_t^\pi)\right]\leq\hrho_t+\epsilon_t, \quad t\in[T],\\
&\sum_{t=1}^T\E\left[x_t^\pi G(x_t^\pi)\right]\leq B.
\end{aligned}
\end{equation}

When $\bepi=\bm 0$, the global constraint is implied by the period caps because $\sum_{t=1}^T\hrho_t=B$, so \eqref{eq:optplan_error} reduces to the strict comparator $\VPLAN(\bhrho)$. For a general allowance vector, the global constraint must be retained: the enlarged period caps may sum to more than $B$, but the relaxed comparator should still obey the same aggregate expected budget as $\VOPT$. Therefore,
\[
\VPLAN(\bhrho)\leq\VPLAN(\bhrho,\bepi)\leq\VOPT.
\]

The evaluator chooses $\epsilon_t$ to specify how much the benchmark may exceed the predicted plan in period $t$. The allowance is neither forecast error nor an input to the online policy. This is different from $\PlanErr$, which is an ex-post measure of the distance between the predicted plan and the ideal allocation. Thus, the allowance term below measures the cost of enlarging the comparator, whereas the plan-error term in Section~\ref{sec:known_F} measures the economic cost of an inaccurate prediction.

\begin{thrm}[Relaxed-plan comparator gap]\label{thrm:discussion2}
Consider Algorithm~\ref{alg:1} with a feasible plan $\bhrho$, step size $\eta=T^{-1/2}$, and initial dual variable $0\leq\mu_1\leq b/a+b$. For every nonnegative allowance vector $\bepi$,
\begin{equation}\label{eq:relaxed-plan-bound}
\VPLAN(\bhrho,\bepi)-\Vpi \leq O(\sqrt{T}) + \frac{b}{a}\sum_{t=1}^T\epsilon_t,
\end{equation}
where the hidden constant in the first term depends only on $a$ and $b$.
\end{thrm}

The bound separates the effect of relaxing the benchmark from the cost of online learning. Relative to the strict comparator, a relaxed comparator can use at most $\epsilon_t$ additional expected expenditure in period $t$. The same reward-to-expenditure argument used in Proposition~\ref{prop:plan-sensitivity} shows that removing this extra expenditure reduces expected reward by at most $b/a$ per unit. What remains is the fixed-plan execution gap controlled by Theorem~\ref{thrm:discussion1}. The allowance vector enlarges only the evaluator's comparator. Algorithm~\ref{alg:1} continues to use $\bhrho$ as a soft target and the remaining budget as its pathwise constraint. Consequently, \eqref{eq:relaxed-plan-bound} controls a signed, one-sided comparator gap.

The relaxed comparator recovers the global benchmark whenever its period caps admit a globally optimal expected-expenditure profile. The ideal allocation in \eqref{def:rho1} provides an exact recovery choice. Define
\begin{equation}\label{eq:ideal-allowance}
\epsilon_t^{\rho}\triangleq(\rho_t-\hrho_t)^+,
\quad t\in[T].
\end{equation}

\begin{coro}[Recovery of the global comparator]\label{coro:relaxed-global}
For $\bepi^{\rho}=(\epsilon_1^{\rho},\ldots,\epsilon_T^{\rho})$,
\begin{equation}\label{eq:relaxed-equals-global}
\VPLAN(\bhrho,\bepi^{\rho})=\VOPT,
\quad
\sum_{t=1}^T\epsilon_t^{\rho} = \frac{1}{2}\PlanErr.
\end{equation}
Together with Theorem~\ref{thrm:discussion2}, this choice recovers the informative-regime guarantee in Theorem~\ref{thrm:informative}.
\end{coro}

Thus, the strict- and relaxed-comparator routes are complementary representations of the same global performance guarantee. Indeed, the period cap generated by \eqref{eq:ideal-allowance} is $\hrho_t+\epsilon_t^{\rho}=\max\{\hrho_t,\rho_t\}$. The globally optimal benchmark used to construct the ideal allocation spends at most $\rho_t$ in period $t$, so it is feasible for the relaxed comparator. Conversely, every policy feasible for the relaxed comparator satisfies the retained global expected-budget constraint and is therefore feasible for $\VOPT$. This establishes the equality in \eqref{eq:relaxed-equals-global} and clarifies the relationship among the global, strict-plan, and relaxed-plan comparators.

%% file: major_revision/6_numerical.tex
The numerical study asks three questions: how performance scales with the horizon, how informative plans compare with the Uniform Plan and Standard Pacing, and how plan misspecification affects performance. All three experiments use the same nonstationary environments and common auction paths across methods.

\input{experiments/experiment1_values.tex}
\input{experiments/experiment2_values.tex}

\subsection{Experimental Setup}
\label{subsec:numerical-setup}

Let $u_t=(t-1/2)/T$, $\kappa_T=(100/T)^{1/4}$, and $v_t\sim\operatorname{Unif}[\nu_t-0.1,\nu_t+0.1]$. The three value paths are
\begin{equation}\label{eq:continuous-center-paths}
\begin{aligned}
\text{Gradual Change:}\quad
    &\nu_t=0.5+0.2\kappa_T(1-2u_t),\\
\text{Seasonal Change:}\quad
    &\nu_t=0.5+0.2\kappa_T\cos(2\pi u_t),\\
\text{Abrupt Change:}\quad
    &\nu_t=0.3\mathbf 1\{u_t<0.5\}
       +0.7\mathbf 1\{u_t\geq0.5\}.
\end{aligned}
\end{equation}
Independently, $m_t\stackrel{\mathrm{iid}}{\sim}\operatorname{Unif}[0.1,0.5]$. We set $a=0.1$, $b=1$, $B=0.1T$, and $T\in\{100,200,\ldots,1000\}$. Gradual and Seasonal Change satisfy $\mathcal W_T=\Theta(T^{3/4})=o(T)$; increasing $T$ also reduces their variation amplitude. Abrupt Change retains fixed amplitude and $\mathcal W_T=\Theta(T)$ as a harder stress test.

The distribution-informed expected-budget benchmark and its spending path are available in closed form. With $d\triangleq0.5-0.1=0.4$, $\beta\triangleq B/T=0.1$, $M_{2,t}=\E[v_t^2]=\nu_t^2+0.1^2/3$, and $\theta=1+\mu$, they are
\begin{equation}\label{eq:continuous-oracle}
x^*(v;\theta)=\frac{0.1+v/\theta}{2},\quad
\theta^*=\sqrt{\frac{T^{-1}\sum_{t=1}^T M_{2,t}}{0.1^2+4d\beta}},\quad
\rho_t=\frac{M_{2,t}/(\theta^*)^2-0.1^2}{4d}.
\end{equation}

For every reported scenario and horizon, the optimizer in \eqref{eq:continuous-oracle} lies in the competing-bid support, so the displayed interior formula applies without truncation. The code evaluates the corresponding benchmark value $\VOPT$ exactly and verifies $\sum_t\rho_t=B$. This benchmark knows the experimental distributions but does not observe future auction realizations; it is not a sample-path clairvoyant optimum.

We compare four methods. The \emph{Ideal Plan} uses $(\rho_t)_{t=1}^T$ computed from the true experimental distributions and serves as an evaluation reference, not information available to a practical bidder. The \emph{Imperfect Plan} applies the same planner to $\widehat\nu_t=0.5+0.6(\nu_t-0.5)$, retaining 60\% of the variation in the value-center path. This is a controlled perturbation, not a measure of forecast accuracy. The \emph{Uniform Plan} uses $B/T$ in every period. \emph{Standard Pacing} is the full-threshold-feedback method of \citet{wang2023learning}, with a constant spending target and estimated expected-spending update. We retain its published initialization, tie rule, and stopping rule and apply it to the same nonstationary paths; its existing theorem assumes i.i.d. values.

Standard Pacing shares the auction format, threshold feedback, empirical-CDF bid optimization, and hard-budget objective. The Uniform Plan is the matched ablation: it keeps our controller fixed and removes only the nonuniform timing signal.

The Ideal Plan, Imperfect Plan, and Uniform Plan share $\mu_1=0.2$ and $\eta=T^{-1/2}$ across scenarios and horizons. The common initial price is obtained by rounding the shadow price of the stationary reference model $v\sim\mathrm{Uniform}[0.4,0.6]$, $m\sim\mathrm{Uniform}[0.1,0.5]$, and $B/T=0.1$ to one decimal place; Table~\ref{tab:initialization-sensitivity} reports sensitivity to this choice. Standard Pacing retains its source initialization and uses the fixed step size $T^{-1/2}$. From period $2$ onward, empirical-CDF maximization is solved exactly by evaluating zero and the previously revealed threshold jumps. In the first period, the three plan-driven methods use the prescribed initialization $\widehat G_1(x)=\mathbf 1\{x\geq a\}$.

Each $(\text{scenario},T)$ cell contains $N=400$ independent paths shared by all methods. Let $V_i^\pi$ be policy $\pi$'s realized total surplus on path $i$, and define
\begin{equation}\label{eq:continuous-relative-gap}
\operatorname{RelativeError}_i(\pi)
=100\frac{\VOPT-V_i^\pi}{\VOPT}.
\end{equation}

Relative Error is the percentage surplus gap to the expected-budget benchmark; lower is better, and individual realizations can be negative. We report means with pointwise two-sided 95\% Student-$t$ intervals and paired method differences. No simulated policy overspends.

\subsection{Experiment 1: Performance over Longer Horizons}
\label{subsec:numerical-horizons}

\begin{figure}[t]
    \centering
    \includegraphics[width=\textwidth]{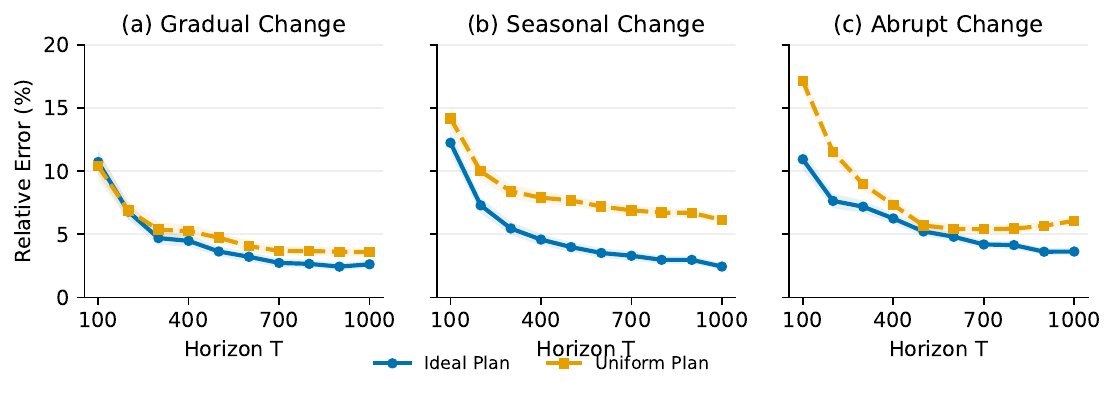}
    \caption{Relative Error over longer horizons for the Ideal and Uniform Plans. Curves show means with shaded pointwise 95\% Student-$t$ intervals over $N=400$ common paths per horizon. Lower is better. Gradual and Seasonal Change have $\mathcal W_T=o(T)$; Abrupt Change has fixed amplitude.}
    \label{fig:horizon-ideal-uniform}
\end{figure}

Figure~\ref{fig:horizon-ideal-uniform} compares exact timing information with no timing information. Relative Error generally falls with the horizon. At the common initialization $\mu_1=0.2$, the Ideal Plan is lower throughout Seasonal and Abrupt Change and, except at $T=100$, in Gradual Change. Table~\ref{tab:initialization-sensitivity} gives the paired $T=1000$ differences and varies only the initial price on the same paths. The Ideal Plan is lower across all five initial prices in Seasonal Change; in Gradual and Abrupt Change, the Uniform Plan is lower at $\mu_1=0$ and $0.1$, whereas the Ideal Plan is lower at $0.2$, $0.3$, and $0.4$. Fluid optimality therefore does not imply uniform finite-horizon dominance of the online controller.

\begin{table}[t]
\centering
\caption{Initialization sensitivity at $T=1000$. Entries give paired Uniform-minus-Ideal Relative Error in percentage points: mean $\pm$ pointwise 95\% Student-$t$ confidence half-width over the same $N=400$ paths. Positive values favor the Ideal Plan.}
\label{tab:initialization-sensitivity}
\small
\setlength{\tabcolsep}{4pt}
\begin{tabular}{@{}lccccc@{}}
\toprule
Initial price $\mu_1$ & $0.0$ & $0.1$ & $0.2$ & $0.3$ & $0.4$\\
\midrule
Gradual Change & $-3.53\pm0.14$ & $-1.80\pm0.17$ & $0.96\pm0.15$ & $2.27\pm0.12$ & $2.48\pm0.12$\\
Seasonal Change & $4.27\pm0.14$ & $3.94\pm0.15$ & $3.69\pm0.15$ & $1.91\pm0.18$ & $0.64\pm0.13$\\
Abrupt Change & $-2.86\pm0.21$ & $-0.22\pm0.21$ & $2.44\pm0.20$ & $4.62\pm0.19$ & $4.72\pm0.19$\\
\bottomrule
\end{tabular}
\end{table}

\subsection{Experiment 2: Comparison with Alternative Pacing Methods}
\label{subsec:numerical-methods}

\begin{figure}[t]
    \centering
    \includegraphics[width=\textwidth]{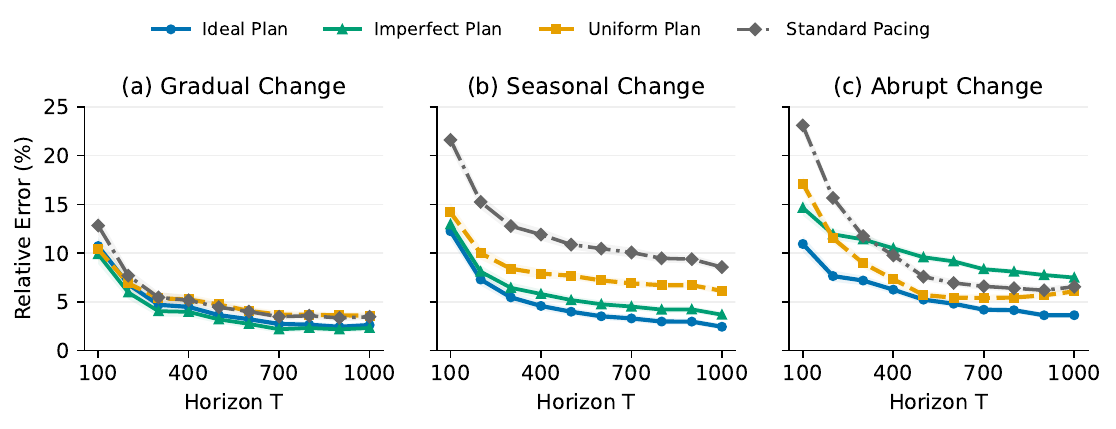}
    \caption{Comparison with alternative pacing methods over the full horizon grid. Curves show mean Relative Error with shaded pointwise 95\% Student-$t$ intervals over $N=400$ common paths per horizon. Standard Pacing follows the full-threshold-feedback algorithm of \citet{wang2023learning}.}
    \label{fig:methods-over-horizon}
\end{figure}

Figure~\ref{fig:methods-over-horizon} adds the Imperfect Plan and Standard Pacing. With the common reference initialization, the Ideal Plan has lower mean error than Standard Pacing at every tested horizon. The Imperfect Plan also improves on both timing-agnostic alternatives throughout Gradual and Seasonal Change. Under Abrupt Change, however, it has higher error than both from $T=400$ onward. Here the perturbation preserves the change point but attenuates its amplitude.

At $T=1000$, the paired Imperfect-minus-Ideal differences are $\GradualImperfectIdealMean$, $\SeasonalImperfectIdealMean$, and $\AbruptImperfectIdealMean$ percentage points for Gradual, Seasonal, and Abrupt Change; the Standard-Pacing-minus-Ideal differences are $\GradualStandardIdealMean$, $\SeasonalStandardIdealMean$, and $\AbruptStandardIdealMean$. Negative values favor the comparison method. Comparisons with Standard Pacing do not isolate plan information alone because its update, initialization, and stopping rule also differ.

\subsection{Experiment 3: Sensitivity to Plan Misspecification}
\label{subsec:numerical-plan-errors}

At $T=1000$, we vary the forecast path $\widehat\nu_t(\alpha)=0.5+\alpha(\nu_t-0.5)$ and apply the same planner used above. The reduced-information series uses $\alpha\in\{1,0.75,0.5,0.25,0\}$, and the reversal series uses $\alpha\in\{0,-0.5,-1\}$. Every resulting plan is nonnegative and sums to $B$. For display, we use the dimensionless \emph{Normalized Plan Error}
\begin{equation}\label{eq:numerical-plan-error}
\NPE(\bhrho) \triangleq
100\,\frac{\PlanErr}{B}=100\,\frac{\sum_{t=1}^T|\rho_t-\widehat\rho_t|}{B}.
\end{equation}
The theoretical Plan Error $\PlanErr$ remains the unnormalized $\ell_1$ quantity used in Theorem~\ref{thrm:informative}.

\begin{figure}[t]
    \centering
    \includegraphics[width=\textwidth]{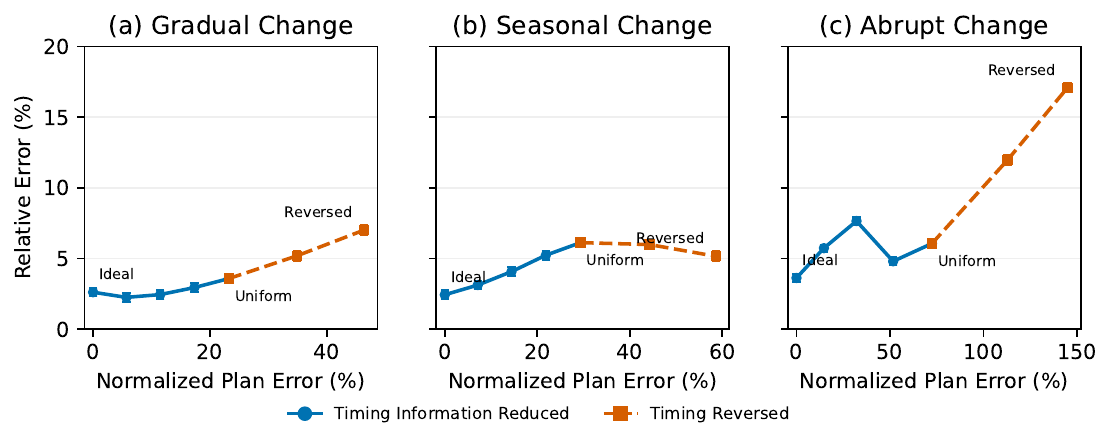}
    \caption{Plan misspecification at $T=1000$. The solid segment removes correct timing information from the Ideal Plan to the Uniform Plan; the dashed segment reverses it. The horizontal axis shows post-construction Normalized Plan Error, $\NPE(\bhrho)$. Vertical bars are pointwise 95\% Student-$t$ intervals over the same $N=400$ paths. Panels have different horizontal ranges because the perturbation produces different errors across environments.}
    \label{fig:plan-misspecification}
\end{figure}

Figure~\ref{fig:plan-misspecification} shows substantial losses from severe timing errors, especially full reversal under Abrupt Change, but the relationship is not monotone. In Gradual Change, $\alpha=0.75$ slightly outperforms the Ideal Plan; in Seasonal Change, full reversal outperforms the Uniform Plan but not the Ideal Plan. Under Abrupt Change, reducing $\alpha$ from $0.5$ to $0.25$ increases Normalized Plan Error from 32.10\% to 51.68\%, while mean Relative Error falls from 7.65\% to 4.81\%. The paired change is $-2.85$ percentage points (95\% interval $[-2.99,-2.70]$).

Theorem~\ref{thrm:informative} gives a worst-case upper bound in $\PlanErr$, not a monotone ranking of individual plans. The Ideal Plan is optimal for the fluid benchmark, not necessarily for finite-horizon execution by Algorithm~\ref{alg:1}. The experiments do not isolate a unique mechanism for the reversals. Normalized Plan Error can reach 200\% because moving budget changes both its source and destination coordinates.

\subsection{Summary and Scope}
\label{subsec:numerical-summary}

The experiments distinguish the benefit of timing information from its finite-horizon execution. Informative plans often reduce Relative Error under the common reference initialization, but the gains depend on the environment, plan accuracy, and initial price. Severe timing errors can be costly, and $\ell_1$ plan error does not order all finite-horizon outcomes. This study isolates value-side nonstationarity under stationary competition and full-threshold feedback. Its plan perturbations evaluate the modular interface in Section~\ref{subsec:plan-construction}, rather than a trained forecasting system.

%% file: experiments/experiment1_values.tex
\newcommand{\GradualUniformIdealMean}{0.96}
\newcommand{\GradualUniformIdealLow}{0.82}
\newcommand{\GradualUniformIdealHigh}{1.11}
\newcommand{\SeasonalUniformIdealMean}{3.69}
\newcommand{\SeasonalUniformIdealLow}{3.53}
\newcommand{\SeasonalUniformIdealHigh}{3.84}
\newcommand{\AbruptUniformIdealMean}{2.44}
\newcommand{\AbruptUniformIdealLow}{2.24}
\newcommand{\AbruptUniformIdealHigh}{2.64}

%% file: experiments/experiment2_values.tex
\newcommand{\GradualImperfectIdealMean}{-0.31}
\newcommand{\SeasonalImperfectIdealMean}{1.23}
\newcommand{\AbruptImperfectIdealMean}{3.86}
\newcommand{\GradualStandardIdealMean}{0.85}
\newcommand{\SeasonalStandardIdealMean}{6.12}
\newcommand{\AbruptStandardIdealMean}{2.92}

%% file: major_revision/proof_3.tex
\section{Proofs for the Uninformative Regime}\label{app:proof-uninformative}

This appendix proves the upper and lower bounds for the uninformative regime. We first collect the reusable ingredients---bounded projected-dual iterates, expected cumulative empirical-CDF control, and projected online-gradient analysis---and then prove the two model-specific steps needed here: comparing the time-varying private-value distributions through the Wasserstein metric and coupling the auxiliary process to the implemented hard-budget policy.

\subsection{Auxiliary Results}\label{appx:subsec:prep-unknown-F}

\begin{lemma}[Bounded dual iterates]\label{lemma:bdd-mu}
Let $M\triangleq b/a+b$. Suppose $0\leq\mu_1\leq M$, $0<\eta\leq1$, and the plan targets satisfy $\hrho_t\geq0$ for every $t$. Under the auxiliary, untruncated version of Algorithm~\ref{alg:1},
\[
0\leq\muR_t\leq M, \quad t\in[T+1].
\]
\end{lemma}

\begin{proof}{Proof.}
Nonnegativity follows immediately from projection onto $\mathbb R_+$. We prove the upper bound by induction, using the global convention that every bid is the smallest maximizer.
Because every competing bid lies in $[a,b]$, both the empirical CDF $\widehat G_t$ and its initialization satisfy $\widehat G_t(x)=0$ for $x<a$. If $\muR_t\geq b/a$, then, for every $x\geq a$,
\[
v_t-(1+\muR_t)x \leq b-(1+b/a)a =-a<0.
\]

Bids in $[0,a)$ have empirical winning probability zero, so the maximum empirical objective value is zero and the tie-breaking rule selects $\xR_t=0$. Hence $\zR_t=0$, the update direction is nonnegative, and $\muR_{t+1}\leq\muR_t\leq M$.

If instead $\muR_t<b/a$, then $0\leq\zR_t\leq b$, and therefore the update direction is bounded below by $-b$. Thus,
\[
\muR_{t+1} =\bigl(\muR_t-\eta(\hrho_t-\zR_t)\bigr)^+ \leq \muR_t+\eta b \leq b/a+b=M.
\]
The induction is complete.
\Halmos
\end{proof}

\begin{lemma}[Dvoretzky--Kiefer--Wolfowitz--Massart inequality]\label{lemma:DKW_ineq}
For $t\geq2$, let $\widehat G_t$ be the empirical CDF of $t-1$ independent observations with common CDF $G$. Then
\[
\Pr\left(\sup_x|\widehat G_t(x)-G(x)|>u\right)
\leq 2e^{-2(t-1)u^2},\quad u>0.
\]
\end{lemma}

Let
\[
\operatorname{err}_G(t) \triangleq \sup_{x\in[0,b]}|\widehat G_t(x)-G(x)|.
\]

For $t\geq2$, integrating the tail bound in Lemma~\ref{lemma:DKW_ineq} gives
\[
\begin{aligned}
\E[\operatorname{err}_G(t)]
&=\int_0^1\Pr\bigl(\operatorname{err}_G(t)>u\bigr)\,\mathrm du\\
&\leq\int_0^\infty2e^{-2(t-1)u^2}\,\mathrm du
=\sqrt{\frac{\pi}{2(t-1)}}.
\end{aligned}
\]

Since $\operatorname{err}_G(1)\leq1$, linearity of expectation yields
\begin{equation}\label{eq:est-err}
\begin{aligned}
\E\left[\sum_{t=1}^T\operatorname{err}_G(t)\right]
&\leq1+\sqrt{\frac{\pi}{2}}\sum_{n=1}^{T-1}\frac{1}{\sqrt n}\\
&\leq1+\sqrt{2\pi(T-1)}
=O(\sqrt T).
\end{aligned}
\end{equation}

For a private-value distribution $F$, define
\begin{equation}\label{eq:def-L}
L(\mu,F) \triangleq \E_{v\sim F}\left[ \max_{x\in[0,b]}(v-(1+\mu)x)G(x) \right].
\end{equation}

For every distribution $F$, the map $\mu\mapsto L(\mu,F)$ is convex because it is the expectation of a pointwise maximum of affine functions of $\mu$.

\begin{lemma}[Wasserstein stability]\label{lemma:W-dist}
For every $\mu\geq0$ and all distributions $F,F'$ on $[0,b]$,
\[
|L(\mu,F)-L(\mu,F')| \leq\mathcal W_1(F,F').
\]
\end{lemma}

\begin{proof}{Proof.}
Recall the value function
\[
\phi(v,\mu)=\max_{x\in[0,b]}(v-(1+\mu)x)G(x).
\]

For fixed $\mu$, this function is one-Lipschitz in $v$: for each fixed $x$, the corresponding affine function has slope $G(x)\in[0,1]$, and taking a pointwise maximum preserves the one-Lipschitz property. Hence, for any coupling $(V,V')$ of $F$ and $F'$,
\[
|L(\mu,F)-L(\mu,F')| =|\E[\phi(V,\mu)-\phi(V',\mu)]| \leq\E|V-V'|.
\]
Taking the infimum over all couplings proves the claim.
\Halmos
\end{proof}

\subsection{Proof of Theorem~\ref{thrm:noninformative}}\label{proof3.1}

\begin{proof}{Proof.}
Set $\bar q\triangleq B/T$, the common coordinate of the uniform plan. Let $\mathcal H_{t-1}$ denote the information available before observing the period-$t$ private value. The empirical CDF $\widehat G_t$ and the auxiliary dual variable $\muR_t$ are $\mathcal H_{t-1}$-measurable, while $v_t\sim F_t$ and $m_t\sim G$ are independent of $\mathcal H_{t-1}$ and of each other.

Consider an auxiliary process that uses the same primitive randomness and the same tie-breaking rule as Algorithm~\ref{alg:1}, but does not truncate its target bid by the remaining budget. Specifically, let
\[
\xR_t\in \argmax_{x\in[0,b]} (v_t-(1+\muR_t)x)\widehat G_t(x),
\quad
\zR_t\triangleq\xR_t\mathbf 1\{\xR_t\geq m_t\},
\]
where $\xR_t$ is the smallest maximizer. Update $\muR_{t+1} =\bigl(\muR_t-\eta(\bar q-\zR_t)\bigr)^+$, with $\muR_1=\mu_1$. The auxiliary process is only for analysis; it may spend more than $B$.

\noindent\textbf{Coupling with the hard-budget policy.}
Before the first period in which the auxiliary process wins with a target bid exceeding the online policy's remaining budget, the two processes have identical realized expenditures and dual states. An unsafe target bid that loses causes no divergence: both systems spend zero and make the same dual update. The tie-breaking rule also ensures that a selected positive bid never exceeds the current private value, so every realized surplus is nonnegative.

Indeed, bidding zero gives empirical objective value zero. If $x>v_t$ and $\widehat G_t(x)>0$, then $(v_t-(1+\muR_t)x)\widehat G_t(x)<0$; if $\widehat G_t(x)=0$, the smallest-maximizer convention selects zero instead. Thus every selected positive bid satisfies $x\leq v_t$, and its realized surplus is nonnegative.

Let $S_T^R\triangleq\sum_{t=1}^T\zR_t$ and write the implemented policy's period-$t$ realized surplus as
\[
r_t^\pi \triangleq (v_t-x_t^\pi)\mathbf 1\{x_t^\pi\geq m_t\}.
\]

If no unsafe auxiliary win occurs, the two reward streams are identical. Otherwise, let $\tau$ be the first such period and let $P_{\tau-1}=\sum_{t<\tau}\zR_t$. The processes agree before $\tau$, and the definition of an unsafe win gives
\[
P_{\tau-1}+\zR_\tau>B.
\]
Let $N_\tau$ be the number of auxiliary wins after period $\tau$. Every such win costs at least $a$, so
\[
aN_\tau\leq \sum_{t>\tau}\zR_t
\leq S_T^R-B.
\]
The first inequality uses the support condition $m_t\geq a$; the second uses $P_{\tau-1}+\zR_\tau>B$. Each auxiliary win yields at most $b$ surplus, while the implemented policy earns nonnegative surplus after the processes diverge. Hence the auxiliary process can exceed the implemented policy's reward by at most $b(1+N_\tau)$. In both cases, the following pathwise comparison holds:
\[
\sum_{t=1}^T r_t^\pi \geq \sum_{t=1}^T(v_t-\xR_t)\mathbf 1\{\xR_t\geq m_t\}-b\left(\frac{(S_T^R-B)^+}{a}+1\right).
\]

Moreover, conditional on $\mathcal H_{t-1}$ and $v_t$, the bid $\xR_t$ is independent of $m_t$, and hence
\[
\E\left[(v_t-\xR_t)\mathbf 1\{\xR_t\geq m_t\}\,\middle|\, \mathcal H_{t-1},v_t\right] =(v_t-\xR_t)G(\xR_t).
\]
Taking expectations and rearranging gives
\begin{equation}\label{eq:tmp1}
\begin{aligned}
\VOPT-\Vpi
&\leq
\underbrace{
\VOPT-\E\left[\sum_{t=1}^T(v_t-\xR_t)G(\xR_t)\right]
}_{A_T}\\
&\quad+
\underbrace{
b\,\E\left[\frac{(S_T^R-B)^+}{a}+1\right]
}_{P_T}.
\end{aligned}
\end{equation}
We bound the auxiliary-process loss $A_T$ and the hard-budget correction $P_T$ separately.

\subsubsection{Auxiliary-process loss.}\label{appx:subsec:term-a}

Let $M\triangleq b/a+b$. By Lemma~\ref{lemma:bdd-mu}, $0\leq\muR_t\leq M$. For the current period, define
\[
h_t(x)\triangleq v_t-(1+\muR_t)x
\]
and let
\[
x_t^\star\in\argmax_{x\in[0,b]}h_t(x)G(x)
\]
be the smallest population maximizer. Since $\xR_t$ maximizes $h_t(x)\widehat G_t(x)$,
\[
h_t(\xR_t)\widehat G_t(\xR_t)
\geq h_t(x_t^\star)\widehat G_t(x_t^\star).
\]
Also, for every $x\in[0,b]$, $|h_t(x)| \leq b+(1+M)b \triangleq K$. Adding and subtracting $\widehat G_t$ at $\xR_t$ and $x_t^\star$ therefore yields
\[
h_t(\xR_t)G(\xR_t) \geq h_t(x_t^\star)G(x_t^\star) -2K\operatorname{err}_G(t).
\]

Since $(v_t-\xR_t)G(\xR_t) =h_t(\xR_t)G(\xR_t) +\muR_t\xR_tG(\xR_t)$, we obtain
\begin{equation}\label{eq:empirical-bid-comparison}
(v_t-\xR_t)G(\xR_t) \geq \phi(v_t,\muR_t) +\muR_t\xR_tG(\xR_t) -C_G\operatorname{err}_G(t),
\end{equation}
where $C_G\triangleq 2K$ depends only on $a$ and $b$.
Because $\muR_t$ is determined before observing $v_t$, we know that
\[
\E[\phi(v_t,\muR_t)\mid\mathcal H_{t-1}]=L(\muR_t,F_t).
\]

Likewise, conditioning additionally on $v_t$ gives $\E[\zR_t\mid\mathcal H_{t-1},v_t]=\xR_tG(\xR_t)$. Using these identities in \eqref{eq:empirical-bid-comparison}, and adding and subtracting $\muR_t\bar q$, gives
\begin{equation}\label{eq:tmp7}
\begin{aligned}
A_T
&\leq
\underbrace{\VOPT-\E\left[\sum_{t=1}^T\{L(\muR_t,F_t)+\muR_t\bar q\}\right]}_{C_T}\\
&\quad+
\underbrace{\E\left[\sum_{t=1}^T\muR_t(\bar q-\zR_t)\right]}_{D_T}
+ C_G\E\left[\sum_{t=1}^T\operatorname{err}_G(t)\right].
\end{aligned}
\end{equation}

We first bound $C_T$. Put $\bar\mu \triangleq \frac1T\sum_{t=1}^T\muR_t$. For each realization of the auxiliary history, Lemma~\ref{lemma:W-dist} and the convexity of $L(\cdot,\bar F_T)$ imply
\begin{align}
\sum_{t=1}^T \{L(\muR_t,F_t)+\muR_t\bar q\}
&\geq
\sum_{t=1}^T\{L(\muR_t,\bar F_T)+\muR_t\bar q\}-\mathcal W_T\notag\\
&\geq
T\{L(\bar\mu,\bar F_T)+\bar\mu\bar q\}-\mathcal W_T.
\label{eq:WT-lower-comparison}
\end{align}

On the other hand, weak duality holds for every realized value of $\bar\mu$:
\begin{align}
\VOPT
&\leq\VLR(\bar\mu)
=\sum_{t=1}^T\{L(\bar\mu,F_t)+\bar\mu\bar q\}\notag\\
&\leq
T\{L(\bar\mu,\bar F_T)+\bar\mu\bar q\}+\mathcal W_T,
\label{eq:WT-upper-comparison}
\end{align}
where the last inequality again uses Lemma~\ref{lemma:W-dist}. Taking expectations in \eqref{eq:WT-lower-comparison}--\eqref{eq:WT-upper-comparison} and subtracting gives
\begin{equation}\label{eq:tmp5}
C_T \leq 2\sum_{t=1}^T\mathcal W_1(F_t,\bar F_T) =2\mathcal W_T.
\end{equation}

This is the distribution-shift step specific to the uninformative regime.

For $D_T$, nonexpansiveness of projection and $|\bar q-\zR_t|\leq b$ imply
\[
(\muR_{t+1})^2 \leq (\muR_t)^2 +\eta^2b^2 -2\eta\muR_t(\bar q-\zR_t).
\]

Summing over $t$, taking expectations, and using $\muR_{T+1}\geq0$ yields
\begin{equation}\label{eq:tmp6}
D_T \leq \frac{\eta b^2T}{2} +\frac{\mu_1^2}{2\eta} =O(\sqrt T)
\end{equation}
for $\eta=T^{-1/2}$. Combining \eqref{eq:est-err}, \eqref{eq:tmp5}, and \eqref{eq:tmp6} in \eqref{eq:tmp7} gives
\[
A_T \leq O(\sqrt{T}) +2\mathcal W_T.
\]

\subsubsection{Hard-budget correction.}\label{appx:subsec:term-b}

The projected update also satisfies $\muR_{t+1} \geq \muR_t-\eta(\bar q-\zR_t)$. Rearranging, summing over $t$, and using $T\bar q=B$ gives
\[
S_T^R-B =\sum_{t=1}^T(\zR_t-\bar q) \leq \frac{\muR_{T+1}-\mu_1}{\eta} \leq \frac{M}{\eta},
\]
where the final inequality follows from Lemma~\ref{lemma:bdd-mu} and $\mu_1\geq0$. Hence,
\[
(S_T^R-B)^+ \leq\frac{M}{\eta}
\]
and therefore
\[
P_T \leq b\left(\frac{M}{a\eta}+1\right) =O(\sqrt T)
\]
for $\eta=T^{-1/2}$. Substituting the bounds on $A_T$ and $P_T$ into \eqref{eq:tmp1} proves
\[
\VOPT-\Vpi \leq O(\sqrt{T}) +O(\mathcal W_T).
\Halmos
\]
\end{proof}

\subsection{Proof of Proposition~\ref{prop:lbd-W}}

\begin{proof}{Proof.}
We construct two environments that are identical during an initial block but differ in whether the remaining opportunities become more or less valuable.

\noindent\textbf{Step 1: two indistinguishable instances.}
Let
\[
n\triangleq\left\lfloor\frac{T}{2}\right\rfloor,
\quad
\ell\triangleq T-n=\left\lceil\frac{T}{2}\right\rceil.
\]

The first block contains $n$ periods and the second block contains $\ell$ periods. Define
\[
\bar w_T\triangleq\frac{n\ell}{T}(b-a),
\quad
\widetilde w\triangleq\min\{w,\bar w_T\},
\]
and set
\[
u\triangleq\frac{a+b}{2},
\quad
\Delta\triangleq\frac{T\widetilde w}{2n\ell}.
\]

Because $\widetilde w\leq\bar w_T$, we have $0\leq\Delta\leq(b-a)/2$, and hence $u-\Delta$ and $u+\Delta$ both belong to $[a,b]$.

Let the competing bid be deterministic, $m_t=a$ for every $t\in[T]$, and set the budget to $B=an$. A bid below $a$ never wins. Every winning bid costs at least $a$ and yields surplus at most $v_t-a$, with equality when the bid is $a$. Thus, any hard-budget-feasible policy can win at most $n$ auctions, and the fluid benchmark attains the values used below by bidding $a$ in the selected periods.
Construct two instances, denoted by $H$ and $L$. In both instances,
\[
F_t^H=F_t^L=\delta_u, \quad t\leq n,
\]
where $\delta_y$ denotes the point mass at $y$. In the second block, let
\[
F_t^H=\delta_{u+\Delta}, \quad
F_t^L=\delta_{u-\Delta}, \quad t>n.
\]

The two instances therefore generate identical values, competing bids, and feedback throughout the first block. They differ only in whether the second-block opportunities become more or less valuable.

\noindent\textbf{Step 2: heterogeneity of the two instances.}
For either instance, the campaign-wide mixture is
\[
\bar F_T^{H/L} = \frac{n}{T}\delta_u + \frac{\ell}{T}\delta_{u\pm\Delta}.
\]

For a first-block period, the Wasserstein distance to this mixture is $\ell\Delta/T$, because a mass $\ell/T$ is located a distance $\Delta$ away. For a second-block period, the distance is $n\Delta/T$. Therefore,
\[
\begin{aligned}
\mathcal W_T
& = n\frac{\ell}{T}\Delta + \ell\frac{n}{T}\Delta\\
& = \frac{2n\ell}{T}\Delta = \widetilde w \leq w.
\end{aligned}
\]
Thus, both instances satisfy the required heterogeneity bound.

\noindent\textbf{Step 3: regret tradeoff between spending early and saving the budget.}
Let $Q$ be the expected number of auctions won by policy $\pi$ during the first $n$ periods. Since the two instances are identical throughout this block, nonanticipativity implies that $Q$ is the same in both.

Write $R_T^J(\pi)\triangleq V_J^{\mathrm{OPT}}-V_J^\pi$ for the regret in instance $J\in\{H,L\}$. In instance $H$, the second-block opportunities are more valuable. The fluid benchmark spends its budget on $n$ of the $\ell\geq n$ second-block auctions and obtains $V_H^{\mathrm{OPT}}=n(u+\Delta-a)$. No feasible policy can do better because it can win at most $n$ auctions and $u+\Delta$ is the largest value. The online policy has at most $Q$ expected first-block wins and at most $n-Q$ expected second-block wins. Hence,
\[
R_T^H(\pi)\geq Q\Delta.
\]

In instance $L$, the first-block opportunities are more valuable. The fluid benchmark uses its budget for the $n$ first-block auctions and obtains $V_L^{\mathrm{OPT}}=n(u-a)$. Because $u-\Delta-a\geq0$, the online policy's expected reward is at most $Q(u-a)+(n-Q)(u-\Delta-a)$, and therefore
\[
R_T^L(\pi)\geq(n-Q)\Delta.
\]

Consequently,
\[
\max\{R_T^H(\pi),R_T^L(\pi)\} \geq
\max\{Q\Delta,(n-Q)\Delta\} \geq \frac{n\Delta}{2}.
\]

If $w\leq\bar w_T$, then $\widetilde w=w$, and
\[
\frac{n\Delta}{2} = \frac{Tw}{4\ell} \geq\frac{w}{4}.
\]

If $w>\bar w_T$, then $\Delta=(b-a)/2$. Since $n\geq T/3$ for every $T\geq2$ and $w\leq T(b-a)$,
\[
\frac{n\Delta}{2} = \frac{n(b-a)}{4} \geq \frac{T(b-a)}{12} \geq \frac{w}{12}.
\]
Thus, in either case, every online policy incurs regret at least $w/12$ on one of the two instances.
\Halmos
\end{proof}

%% file: major_revision/proof_4.tex
\section{Proofs for the Informative Regime}\label{app:proof-informative}

The main paper defines $\VOPT$ and $\VPLAN(\bhrho)$ through the policy class $\Pi_0$. This appendix gives an equivalent periodwise randomized-rule representation that preserves both comparator values while enabling separability, compactness, duality, and sensitivity analysis. The representation is used solely in the proof; Algorithm~\ref{alg:1} continues to operate directly on the supplied plan.

\subsection{Periodwise Representation of the Comparators}\label{app:informative-separable}

We define a randomized period-$t$ rule $\kappa_t$ to map the observed value $v$ to a distribution over bids in $[0,b]$. For such a rule, define its expected reward and expenditure by
\begin{align}
R_t(\kappa_t)
&\triangleq \E_{v_t\sim F_t}\left[\int_{[0,b]}(v_t-x)G(x)\,\kappa_t(\dd x\mid v_t)\right],
\label{eq:period-reward}\\
C_t(\kappa_t)
&\triangleq \E_{v_t\sim F_t}\left[\int_{[0,b]}xG(x)\,\kappa_t(\dd x\mid v_t)\right].
\label{eq:period-consumption}
\end{align}

We first explain why these rules give an equivalent representation of the policy-based comparators. Consider any $\pi\in\Pi_0$. In period $t$, condition on $v_t=v$ and average the distribution of $x_t^\pi$ over the pre-period history and any internal randomization of the policy. The resulting kernel $\kappa_t(\cdot\mid v)$ has the same expected period-$t$ reward and expenditure as $x_t^\pi$. Conversely, any collection $(\kappa_1,\ldots,\kappa_T)$ can be implemented by a nonanticipative policy that, after observing $v_t$, draws its bid from $\kappa_t(\cdot\mid v_t)$.  
This converse is valid because $\Pi_0$ has no pathwise remaining-budget restriction; aggregate or periodwise expected feasibility is imposed only by the displayed comparator constraints. Therefore,
\begin{equation}\label{eq:fluid-separable}
\VOPT =\sup_{\kappa_1,\ldots,\kappa_T} \left\{\sum_{t=1}^T R_t(\kappa_t): \sum_{t=1}^T C_t(\kappa_t)\leq B\right\},
\end{equation}
and the strict plan comparator in \eqref{eq:optplan} satisfies
\begin{equation}\label{eq:plan-separable}
\VPLAN(\bhrho) =\sup_{\kappa_1,\ldots,\kappa_T} \left\{\sum_{t=1}^T R_t(\kappa_t): C_t(\kappa_t)\leq\hrho_t, \ t\in[T]\right\}.
\end{equation}

For $c\geq0$, define the period value function
\begin{equation}\label{eq:period-value}
\Phi_t(c)\triangleq \sup_{\kappa_t}\left\{R_t(\kappa_t):C_t(\kappa_t)\leq c\right\}.
\end{equation}

Because the constraints in \eqref{eq:plan-separable} are separate across periods,
\begin{equation}\label{eq:plan-value-sum}
\VPLAN(\bhrho)=\sum_{t=1}^T\Phi_t(\hrho_t).
\end{equation}

The function $\Phi_t$ is nondecreasing and concave. Additional spending capacity cannot reduce value, while randomizing between two rules produces the corresponding convex combination of their reward--expenditure pairs.

To make attainment and strong duality precise, let
\begin{equation}\label{eq:attainable-pairs}
\mathcal A_t\triangleq
\operatorname{closure}\left\{
\bigl(C_t(\kappa_t),R_t(\kappa_t)\bigr):
\kappa_t\text{ is a randomized period-$t$ rule}
\right\}.
\end{equation}

Each $\mathcal A_t$ is bounded and convex and contains $(0,0)$. We will repeatedly use the following elementary upper bound. For $x<a$, $G(x)=0$. For $x\geq a$ and $v_t\leq b$,
\[
(v_t-x)G(x)\leq bG(x)\leq\frac{b}{a}xG(x).
\]
Averaging over values and bid randomization, and then passing to the closure, gives
\begin{equation}\label{eq:attainable-reward-spend}
r\leq\frac{b}{a}c,
\qquad (c,r)\in\mathcal A_t.
\end{equation}

Passing to the closure does not change any comparator value. Every closure point can be approximated by randomized rules. If its expenditure equals a positive cap and an approximation slightly exceeds that cap, mixing the approximation with $(0,0)$ by a factor tending to one restores feasibility without changing the limit. If the cap is zero, \eqref{eq:attainable-reward-spend} implies that every feasible closure point has reward at most zero, and $(0,0)$ attains zero. The same reasoning applies to the aggregate cap. Thus, the suprema above can be treated as maxima over the compact sets $\mathcal A_t$.

\subsection{Proof of Lemma~\ref{lem:decomposition}}

\begin{proof}{Proof.}
Using \eqref{eq:attainable-pairs}, the global fluid benchmark is the finite-dimensional convex problem
\begin{equation}\label{eq:closed-fluid-problem}
\max\left\{\sum_{t=1}^T r_t: \sum_{t=1}^T c_t\leq B,\ (c_t,r_t)\in\mathcal A_t,\ t\in[T]\right\}.
\end{equation}

Under the maintained assumption $B>0$, the zero-bid solution is strictly feasible for the aggregate budget constraint, so strong Lagrangian duality applies. The dual minimum is also attained: $\VLR$ is continuous in $\mu$ and satisfies $\VLR(\mu)\geq\mu B$, so it is coercive on $\mathbb R_+$.

Moreover, for every $\mu\geq0$,
\begin{equation}\label{eq:period-pair-dual}
\max_{(c,r)\in\mathcal A_t}\{r-\mu c\} =\E_{v_t}\left[\max_{x\in[0,b]}\bigl(v_t-(1+\mu)x\bigr)G(x)\right].
\end{equation}

Indeed, the left side is linear in the distribution of the bid, so randomization cannot improve on the best bid for each realized value. The pointwise maximum is taken over $[0,b]$, including the zero bid; among all maximizers, we select the smallest bid. Consequently, the dual of \eqref{eq:closed-fluid-problem} is exactly $\min_{\mu\geq0}\VLR(\mu)$ from Section~\ref{sec:main_alg}.

Select the minimum-Euclidean-norm vector from the compact convex primal optimal set in \eqref{eq:closed-fluid-problem}, independently of the predicted plan, and let $(c_t^*,r_t^*)\in\mathcal A_t$ be its period-$t$ reward--expenditure pairs. Let $\mu^*$ be the smallest dual minimizer. To verify complementary slackness and periodwise optimality without assuming differentiability, write strong duality as
\begin{align*}
0
&=\mu^*B+\sum_{t=1}^T\max_{(c,r)\in\mathcal A_t}\{r-\mu^*c\}-\sum_{t=1}^T r_t^*\\
&=\mu^*\left(B-\sum_{t=1}^T c_t^*\right)
+\sum_{t=1}^T\left[
\max_{(c,r)\in\mathcal A_t}\{r-\mu^*c\}
-(r_t^*-\mu^*c_t^*)\right].
\end{align*}
Every term on the last line is nonnegative. Hence complementary slackness holds, and each selected optimal pair solves its period problem:
\begin{equation}\label{eq:period-saddle}
(c_t^*,r_t^*) \in \argmax_{(c,r)\in\mathcal A_t}\{r-\mu^*c\},\quad t\in[T].
\end{equation}

Recall that $\rho_t=c_t^*+s_t$, where $s_t\geq0$ is defined in \eqref{eq:slack-completion}. By \eqref{eq:period-pair-dual}, for every $\mu\geq0$,
\begin{align}
D_t^*(\mu)
&=\mu\rho_t+\max_{(c,r)\in\mathcal A_t}\{r-\mu c\}\notag\\
&\geq r_t^*+\mu(\rho_t-c_t^*)=r_t^*+\mu s_t.
\label{eq:period-dual-lower}
\end{align}

We now distinguish the binding and slack cases. If $\mu^*>0$, complementary slackness implies $\sum_t c_t^*=B$. Hence $s_t=0$ for every $t$. Equality holds in \eqref{eq:period-dual-lower} at $\mu=\mu^*$ by \eqref{eq:period-saddle}, so $D_t^*(\mu)\geq r_t^*=D_t^*(\mu^*), \mu\geq0$. Thus, $\mu^*$ minimizes every $D_t^*$.

If $\mu^*=0$, \eqref{eq:period-saddle} says that $(c_t^*,r_t^*)$ maximizes the period reward itself. Therefore, $D_t^*(0)=r_t^*$. Equation~\eqref{eq:period-dual-lower} gives $D_t^*(\mu)\geq r_t^*+\mu s_t\geq D_t^*(0), \mu\geq0$, so $\mu^*=0$ again minimizes every period objective. This argument does not require differentiability and therefore remains valid when $G$ is discrete or the period bid problem has multiple optimal solutions.

Finally, using $\sum_t r_t^*=\VOPT$, $\sum_t s_t=B-\sum_t c_t^*$, and complementary slackness,
\[
\sum_{t=1}^T D_t^*(\mu^*) =\sum_{t=1}^T r_t^*+\mu^*\sum_{t=1}^T s_t =\VOPT.
\]

Since $\mu^*$ minimizes every $D_t^*$ and strong duality gives $\VOPT=\min_{\mu\geq0}\VLR(\mu)$, we conclude that
\[
\VOPT =\min_{\mu\geq0}\VLR(\mu) =\sum_{t=1}^T\min_{\mu_t\geq0}D_t^*(\mu_t).
\]
This proves both statements of the lemma.
\Halmos
\end{proof}

\subsection{Proof of Theorem~\ref{thrm:discussion1}}\label{proof4:strict-plan}

\begin{proof}{Proof.}
The proof follows the auxiliary-process argument used for Theorem~\ref{thrm:noninformative}, but the benchmark comparison is now performed separately for the period caps in \eqref{eq:optplan}. The auxiliary results in Appendix~\ref{appx:subsec:prep-unknown-F} were stated for arbitrary nonnegative plan targets, so Lemma~\ref{lemma:bdd-mu} and the empirical-CDF bound \eqref{eq:est-err} apply without change.

Let $\mathcal H_{t-1}$ denote the information available before observing the period-$t$ private value. Then $\widehat G_t$ and the auxiliary dual variable $\muR_t$ are $\mathcal H_{t-1}$-measurable, while $v_t\sim F_t$ and $m_t\sim G$ are independent of $\mathcal H_{t-1}$ and of each other. Consider the auxiliary process that uses the same primitive randomness as Algorithm~\ref{alg:1}, but does not truncate its target bid by the remaining budget. Every bid maximization below is over $[0,b]$, and among all maximizers we select the smallest bid. Specifically, let
\[
\xR_t\in\argmax_{x\in[0,b]}\bigl(v_t-(1+\muR_t)x\bigr)\widehat G_t(x),
\quad
\zR_t\triangleq\xR_t\mathbf 1\{\xR_t\geq m_t\},
\]
and update $\gR_t\triangleq\hrho_t-\zR_t, \muR_{t+1}=\bigl(\muR_t-\eta\gR_t\bigr)^+, \muR_1=\mu_1$. The auxiliary process is used only for analysis and may spend more than $B$.

As in Appendix~\ref{proof3.1}, let $S_T^R\triangleq\sum_{t=1}^T\zR_t$ and write the implemented policy's period-$t$ realized surplus as
\[
r_t^\pi\triangleq(v_t-x_t^\pi)\mathbf 1\{x_t^\pi\geq m_t\}.
\]

On every sample path,
\[
\sum_{t=1}^T r_t^\pi \geq \sum_{t=1}^T(v_t-\xR_t)\mathbf 1\{\xR_t\geq m_t\}-b\left(1+\frac{(S_T^R-B)^+}{a}\right).
\]

Conditional on $\mathcal H_{t-1}$ and $v_t$, the bid $\xR_t$ is fixed and independent of $m_t$, so
\[
\E\left[(v_t-\xR_t)\mathbf 1\{\xR_t\geq m_t\}\,\middle|\,\mathcal H_{t-1},v_t\right]=(v_t-\xR_t)G(\xR_t).
\]

Taking expectations and rearranging gives
\begin{equation}\label{eq:strict-aux-decomp}
\begin{aligned}
\VPLAN(\bhrho)-\Vpi
&\leq
\underbrace{\VPLAN(\bhrho)-
\E\left[\sum_{t=1}^T(v_t-\xR_t)G(\xR_t)\right]}_{A_T}\\
&\quad+
\underbrace{b\,\E\left[\frac{(S_T^R-B)^+}{a}+1\right]}_{P_T}.
\end{aligned}
\end{equation}

We bound the signed auxiliary comparator term $A_T$ and the hard-budget correction $P_T$ separately.

\subsubsection{Auxiliary-process comparator term.}

Let $M\triangleq b/a+b$. Since $\eta=T^{-1/2}\leq1$, Lemma~\ref{lemma:bdd-mu} gives $0\leq\muR_t\leq M$. We first derive the periodwise weak-duality inequality. For any deterministic $\mu\geq0$ and any rule $\kappa_t$ with $C_t(\kappa_t)\leq\hrho_t$,
\[
R_t(\kappa_t)
\leq R_t(\kappa_t)+\mu\bigl(\hrho_t-C_t(\kappa_t)\bigr)
\leq \mu\hrho_t+\sup_{\kappa_t'}\{R_t(\kappa_t')-\mu C_t(\kappa_t')\}.
\]

Maximizing the left side over all feasible rules and using \eqref{eq:period-pair-dual} yields
\begin{equation}\label{eq:periodwise-weak-plan}
\Phi_t(\hrho_t) \leq \mu\hrho_t+ \E_{v_t}\left[\max_{x\in[0,b]}\bigl(v_t-(1+\mu)x\bigr)G(x)\right].
\end{equation}

This is weak duality for one period. It does not assume that the strict-plan problems in different periods share one optimal multiplier.
For the realized value $v_t$ and auxiliary price $\muR_t$, define $h_t(x)\triangleq v_t-(1+\muR_t)x$ and let $x_t^\star$ be the smallest population maximizer:
\[
x_t^\star\in\argmax_{x\in[0,b]}h_t(x)G(x).
\]

Conditional on $\mathcal H_{t-1}$, the quantity $\muR_t$ is fixed and $v_t$ has distribution $F_t$. We may therefore apply \eqref{eq:periodwise-weak-plan} with $\mu=\muR_t$, average over the history, and use \eqref{eq:plan-value-sum}. Summing over periods gives
\begin{equation}\label{eq:plan-weak-random-mu}
\VPLAN(\bhrho) \leq \E\left[\sum_{t=1}^T(v_t-x_t^\star)G(x_t^\star)+\muR_t\bigl(\hrho_t-x_t^\star G(x_t^\star)\bigr)\right].
\end{equation}

We next compare the empirical maximizer $\xR_t$ with the population maximizer $x_t^\star$. Recall that $\operatorname{err}_G(t)=\sup_{x\in[0,b]}|\widehat G_t(x)-G(x)|$.
For every $x\in[0,b]$, $|h_t(x)|\leq b+(1+M)b\triangleq K$. Since $\xR_t$ maximizes $h_t(x)\widehat G_t(x)$, adding and subtracting $\widehat G_t$ at the two bids gives
\begin{align*}
h_t(\xR_t)G(\xR_t)
&\geq h_t(\xR_t)\widehat G_t(\xR_t)-K\operatorname{err}_G(t)\\
&\geq h_t(x_t^\star)\widehat G_t(x_t^\star)-K\operatorname{err}_G(t)\\
&\geq h_t(x_t^\star)G(x_t^\star)-2K\operatorname{err}_G(t).
\end{align*}

Let $C_G\triangleq2K$. Since
\[
h_t(x)G(x)+\muR_t\hrho_t
=(v_t-x)G(x)+\muR_t\bigl(\hrho_t-xG(x)\bigr),
\]
combining the preceding comparison with \eqref{eq:plan-weak-random-mu} yields
\begin{align}
A_T
&\leq \E\left[\sum_{t=1}^T \muR_t\bigl(\hrho_t-\xR_tG(\xR_t)\bigr)\right] +C_G\E\left[\sum_{t=1}^T\operatorname{err}_G(t)\right]\notag\\
&= \E\left[\sum_{t=1}^T\muR_t\gR_t\right] +O(\sqrt{T}).
\label{eq:strict-est}
\end{align}

For the equality, condition on $\mathcal H_{t-1}$ and $v_t$. Then $\xR_t$ is fixed and $m_t\sim G$ is independent, so
\[
\E[\zR_t\mid\mathcal H_{t-1},v_t] =\xR_tG(\xR_t).
\]

The empirical-CDF term is bounded by \eqref{eq:est-err}.
Because $\hrho_t,\zR_t\in[0,b]$, we have $|\gR_t|\leq b$. Nonexpansiveness of projection onto $\mathbb R_+$ gives
\[
(\muR_{t+1})^2 \leq(\muR_t-\eta\gR_t)^2 =(\muR_t)^2-2\eta\muR_t\gR_t+\eta^2(\gR_t)^2.
\]

Summing over periods and using $\muR_{T+1}\geq0$ yields
\begin{equation}\label{eq:strict-ogd}
\sum_{t=1}^T\muR_t\gR_t \leq\frac{\mu_1^2}{2\eta}+\frac{\eta b^2T}{2}.
\end{equation}

Taking expectations and setting $\eta=T^{-1/2}$ gives
\[
A_T\leq C_A\sqrt{T}
\]
for a constant $C_A$ depending only on $a$ and $b$. The inequality is an upper bound because $A_T$ is a signed comparator gap.
Unlike the proof of Theorem~\ref{thrm:noninformative}, no Wasserstein comparison is needed: \eqref{eq:periodwise-weak-plan} compares the period cap directly with the actual distribution $F_t$.

\subsubsection{Hard-budget correction.}

The proof is the same as Appendix~\ref{appx:subsec:term-b}. We have
\[
P_T\leq b\left(1+\frac{M}{a\eta}\right)=O(\sqrt T)
\]
for $\eta=T^{-1/2}$. Substituting the bounds on $A_T$ and $P_T$ into \eqref{eq:strict-aux-decomp} proves \eqref{eq:strict-plan-learning}.
\Halmos
\end{proof}

\subsection{Proof of Proposition~\ref{prop:plan-sensitivity}}

\begin{proof}{Proof.}
Let $(c_t^*,r_t^*)\in\mathcal A_t$ be the globally optimal period pairs used to construct the ideal allocation. Recall that $c_t^*\leq\rho_t$. By \eqref{eq:attainable-reward-spend}, every $(c,r)\in\mathcal A_t$ satisfies $r\leq(b/a)c$. Moreover, every globally optimal solution has $r_t^*\geq0$: replacing a negative-reward period pair by the feasible zero-bid pair $(0,0)$ would increase reward and reduce expenditure. Hence
\begin{equation}\label{eq:reward-spend-ratio}
0\leq r_t^*\leq\frac{b}{a}c_t^*.
\end{equation}

We now reduce expenditure only in periods where $c_t^*>\hrho_t$. Define
\[
\alpha_t\triangleq
\begin{cases}
1,&c_t^*\leq\hrho_t,\\[2pt]
\hrho_t/c_t^*,&c_t^*>\hrho_t.
\end{cases}
\]

Because $\mathcal A_t$ is convex and contains $(0,0)$, the scaled pair $(\alpha_tc_t^*,\alpha_tr_t^*)$ belongs to $\mathcal A_t$. Moreover,
\[
\alpha_t c_t^*=\min\{c_t^*,\hrho_t\}\leq\hrho_t,
\]
so the collection of scaled pairs is feasible for the strict plan comparator in \eqref{eq:plan-separable}. Therefore, using \eqref{eq:reward-spend-ratio},
\begin{align*}
\VOPT-\VPLAN(\bhrho)
&\leq\sum_{t=1}^T(1-\alpha_t)r_t^*\\
&\leq\frac{b}{a}\sum_{t=1}^T(c_t^*-\hrho_t)^+\\
&\leq\frac{b}{a}\sum_{t=1}^T(\rho_t-\hrho_t)^+\\
&=\frac{b}{2a}\PlanErr.
\end{align*}

The third line uses $c_t^*\leq\rho_t$. For the last line, $\sum_t\rho_t=\sum_t\hrho_t=B$ implies $\sum_t(\rho_t-\hrho_t)^+=\frac12\sum_t|\rho_t-\hrho_t|=\PlanErr/2$. The reverse inequality $\VPLAN(\bhrho)\leq\VOPT$ follows directly from the policy definitions: summing the period caps in \eqref{eq:optplan} gives the aggregate budget $B$. This proves \eqref{eq:global-plan-sensitivity}.
\Halmos
\end{proof}

\subsection{Proof of Theorem~\ref{thrm:informative}}

\begin{proof}{Proof.}
Add and subtract the strict plan comparator:
\[
\VOPT-\Vpi =\bigl(\VOPT-\VPLAN(\bhrho)\bigr) +\bigl(\VPLAN(\bhrho)-\Vpi\bigr).
\]

Proposition~\ref{prop:plan-sensitivity} bounds the first term by $(b/2a)\PlanErr$, and Theorem~\ref{thrm:discussion1} bounds the second term by $O(\sqrt{T})$. The second term is a one-sided comparator gap and may be negative, but both upper bounds remain valid. Hence,
\[
\VOPT-\Vpi
\leq O(\sqrt{T})+\frac{b}{2a}\PlanErr.
\Halmos
\]
\end{proof}

\subsection{Proof of Proposition~\ref{prop:lbd-V}}

\begin{proof}{Proof.}
Fix an integer $1\leq k\leq\lfloor T/2\rfloor$. We construct two environments that provide exactly the same predicted plan and generate identical observations during the first $k$ periods, but whose globally optimal spending decisions are opposite.

\noindent\textbf{Step 1: common auction environment and predicted plan.}
Set $B=ak$ and let the highest competing bid be deterministic, $m_t=a$ for every $t\in[T]$. A bid below $a$ never wins, while any bid at least $a$ wins surely. Replacing a winning bid above $a$ by $a$ preserves the auction outcome, increases surplus, and reduces expenditure. Thus, every useful win may be viewed as costing $a$, and the budget permits at most $k$ wins.

Give the bidder the common predicted plan
\begin{equation}\label{eq:lower-plan}
\hrho_t=
\begin{cases}
a/2,&t\leq2k,\\
0,&t>2k.
\end{cases}
\end{equation}

This plan is feasible because every coordinate lies in $[0,b]$ and $\sum_{t=1}^T\hrho_t=2k(a/2)=ak=B$.

\noindent\textbf{Step 2: two indistinguishable value sequences and their plan errors.}
Let
\[
u\triangleq\frac{a+b}{2},
\quad
\delta\triangleq\frac{b-a}{4}.
\]

In both instances, set $v_t=u$ for $t\leq k$. In the high instance, set $v_t=u+\delta$ for $k<t\leq2k$; in the low instance, set $v_t=u-\delta$ over the same block. Set $v_t=a$ for all $t>2k$. Since
\[
u-\delta=\frac{3a+b}{4},
\quad
u+\delta=\frac{a+3b}{4},
\]
all values lie in $[a,b]$. The two instances are identical through the first $k$ periods, and periods after $2k$ offer zero surplus from a winning bid.
In the high instance, the global benchmark uses its $k$ wins in the second block, where values are $u+\delta$. In the low instance, it uses its $k$ wins in the first block, where values are $u$. The budget binds in both cases. Hence, corresponding ideal allocations are
\[
\rho_t^H=a \cdot\mathbf 1\{k<t\leq2k\},
\quad
\rho_t^L=a \cdot\mathbf 1\{t\leq k\}.
\]

For either ideal allocation, exactly $k$ of the first $2k$ coordinates equal $a$ and the other $k$ equal zero, whereas the predicted plan assigns $a/2$ to every one of these coordinates. Therefore,
\begin{equation}\label{eq:lower-plan-error}
\PlanErr
=k\left|a-\frac a2\right|+k\left|0-\frac a2\right|
=ak.
\end{equation}

\noindent\textbf{Step 3: regret tradeoff between spending early and saving the budget.}
For an arbitrary nonanticipative online policy, let $Q$ be its expected number of wins during the first $k$ periods. Since both instances generate the same values, competing bids, and feedback throughout this block, $Q$ is the same in the two instances.

In the high instance, each first-block win yields at most $u-a$, and the policy can obtain at most $k-Q$ additional expected wins in the second block. Its expected reward is therefore at most $Q(u-a)+(k-Q)(u+\delta-a)$. The global benchmark earns $k(u+\delta-a)$, so the policy's regret is at least $\delta Q$.

In the low instance, the policy's expected reward is at most $Q(u-a)+(k-Q)(u-\delta-a)$, whereas the global benchmark earns $k(u-a)$. Its regret is at least $\delta(k-Q)$. Consequently, one of the two instances satisfies
\[
\max\{\delta Q,\delta(k-Q)\} \geq\frac{\delta k}{2} =\frac{b-a}{8a}\PlanErr,
\]
where the final equality uses \eqref{eq:lower-plan-error}. Since the same predicted plan is supplied in both environments, no nonanticipative policy can avoid this tradeoff uniformly. This proves the proposition.
\Halmos
\end{proof}

%% file: major_revision/proof_5.tex
\section{Proofs for the Relaxed Plan Comparator}
\label{app:proof-plan-benchmark}

This appendix applies the periodwise reward--expenditure representation in Appendix~\ref{app:informative-separable} to compare the relaxed and strict comparators and then recover the global comparator through $\epsilon_t^\rho=(\rho_t-\hrho_t)^+$.

\subsection{Proof of Theorem~\ref{thrm:discussion2}}
\label{proof5.2}

\begin{proof}{Proof.}
By the same averaging argument used for \eqref{eq:fluid-separable} and \eqref{eq:plan-separable}, the policy-based comparator in \eqref{eq:optplan_error} has the equivalent periodwise representation
\begin{equation}\label{eq:relaxed-separable-form}
\begin{aligned}
\VPLAN(\bhrho,\bepi)=\max_{(c_t,r_t)\in\mathcal A_t,\ t\in[T]} \
&\sum_{t=1}^T r_t\\
\text{s.t.} \
&c_t\leq\hrho_t+\epsilon_t,
\quad t\in[T],\\
&\sum_{t=1}^T c_t\leq B,
\end{aligned}
\end{equation}
where the compact attainable sets $\mathcal A_t$ are defined in \eqref{eq:attainable-pairs}. Choose optimal period pairs $(c_t^{\epsilon},r_t^{\epsilon})$ in \eqref{eq:relaxed-separable-form}.

By \eqref{eq:attainable-reward-spend}, every $(c,r)\in\mathcal A_t$ satisfies $r\leq(b/a)c$. Since $(0,0)\in\mathcal A_t$, an optimal solution of \eqref{eq:relaxed-separable-form} cannot have a negative period reward: replacing such a pair by $(0,0)$ would increase total reward and preserve feasibility. Therefore, for every period,
\begin{equation}\label{eq:relaxed-reward-spend-ratio}
0\leq r_t^{\epsilon}\leq\frac{b}{a}c_t^{\epsilon}.
\end{equation}

We now reduce the relaxed solution only where it exceeds the strict cap. Define
\[
\alpha_t\triangleq
\begin{cases}
1,&c_t^{\epsilon}\leq\hrho_t,\\[2pt]
\hrho_t/c_t^{\epsilon},&c_t^{\epsilon}>\hrho_t.
\end{cases}
\]

Because $\mathcal A_t$ is convex and contains $(0,0)$, the scaled pair $(\alpha_t c_t^{\epsilon},\alpha_t r_t^{\epsilon})$ also belongs to $\mathcal A_t$. Moreover,
\[
\alpha_t c_t^{\epsilon}=\min\{c_t^{\epsilon},\hrho_t\}\leq\hrho_t,
\]
so the collection of scaled pairs is feasible for the strict comparator $\VPLAN(\bhrho)$. Hence, using \eqref{eq:relaxed-reward-spend-ratio} and the relaxed cap $c_t^{\epsilon}\leq\hrho_t+\epsilon_t$,
\begin{align}
\VPLAN(\bhrho,\bepi)-\VPLAN(\bhrho)
&\leq\sum_{t=1}^T(1-\alpha_t)r_t^{\epsilon}\notag\\
&\leq\frac{b}{a}\sum_{t=1}^T(c_t^{\epsilon}-\hrho_t)^+\notag\\
&\leq\frac{b}{a}\sum_{t=1}^T\epsilon_t.
\label{eq:relaxed-to-strict-sensitivity}
\end{align}

Finally, add and subtract the strict comparator:
\[
\VPLAN(\bhrho,\bepi)-\Vpi
=\bigl(\VPLAN(\bhrho,\bepi)-\VPLAN(\bhrho)\bigr)
+\bigl(\VPLAN(\bhrho)-\Vpi\bigr).
\]

Equation~\eqref{eq:relaxed-to-strict-sensitivity} bounds the first term, and Theorem~\ref{thrm:discussion1} bounds the second term by $O(\sqrt{T})$. This proves \eqref{eq:relaxed-plan-bound}.
\Halmos
\end{proof}

\subsection{Proof of Corollary~\ref{coro:relaxed-global}}

\begin{proof}{Proof.}
Let $(c_t^*,r_t^*)\in\mathcal A_t$ be the globally optimal period pairs used to construct the ideal allocation in Section~\ref{sec:known_F}. By construction, $c_t^*\leq\rho_t$. For the allowances in \eqref{eq:ideal-allowance},
\[
\hrho_t+\epsilon_t^{\rho} =\max\{\hrho_t,\rho_t\} \geq c_t^*, \quad t\in[T].
\]

The optimal pairs also satisfy $\sum_{t=1}^T c_t^*\leq B$, so they are feasible for the relaxed formulation in \eqref{eq:relaxed-separable-form}. Therefore,
\[
\VPLAN(\bhrho,\bepi^{\rho})\geq\sum_{t=1}^T r_t^*=\VOPT.
\]

The reverse inequality follows directly from \eqref{eq:optplan_error}: every policy feasible for the relaxed comparator satisfies the same aggregate expected-budget constraint as the global benchmark and is therefore feasible for $\VOPT$. Hence $\VPLAN(\bhrho,\bepi^{\rho})=\VOPT$.
Finally, both the ideal and predicted plans sum to $B$. Hence $\sum_{t=1}^T(\rho_t-\hrho_t)=0$, and
\[
\sum_{t=1}^T\epsilon_t^{\rho}
=\sum_{t=1}^T(\rho_t-\hrho_t)^+
=\sum_{t=1}^T(\hrho_t-\rho_t)^+
=\frac{1}{2}\sum_{t=1}^T|\rho_t-\hrho_t|
=\frac{1}{2}\PlanErr.
\]

Applying Theorem~\ref{thrm:discussion2} now gives
\[
\VOPT-\Vpi \leq O(\sqrt{T})+\frac{b}{2a}\PlanErr,
\]
which matches the bound in Theorem~\ref{thrm:informative}.
\Halmos
\end{proof}